\documentclass[11pt]{article}
\newcommand{\rr}{{r}}
\newcommand{\hh}{C}

\newcommand{\bb}{H}
\renewcommand{\leq}{\leqslant}
\renewcommand{\ge}{\geqslant}
\renewcommand{\geq}{\geqslant}
\usepackage{amsmath,amssymb,amsfonts,amsthm,epsfig,bm,xspace,epsfig,latexsym}
\usepackage[dvipsnames,usenames]{color}
\usepackage[pdftex,pagebackref,letterpaper=true,colorlinks=true,pdfpagemode=none,urlcolor=blue,linkcolor=blue,citecolor=BrickRed,pdfstartview=FitH]{hyperref}

\usepackage{fullpage}

\newtheorem{theorem}{Theorem}[section]
\newtheorem{lemma}[theorem]{Lemma}
\newtheorem{claim}[theorem]{Claim}
\newtheorem{proposition}[theorem]{Proposition}

\newtheorem{definition}[theorem]{Definition}

\newtheorem{conjecture}[theorem]{Conjecture}
\newtheorem{observation}[theorem]{Observation}

\newcommand{\ignore}[1]{}

\renewcommand{\Pr}{{\bf Pr}}
\renewcommand{\P}{\mathrm{P}}

\newcommand{\Prx}{\mathop{\bf Pr\/}}
\newcommand{\E}{{\bf E}}
\newcommand{\Ex}{\mathop{\bf E\/}}
\newcommand{\Var}{{\bf Var}}

\newcommand{\R}{\mathbb R}
\newcommand{\N}{\mathbb N}

\newcommand{\kuniquelabelcover}{$k$-\textsc{Unique Label Cover}\xspace}
\newcommand{\klabelcover}{$k$-\textsc{Label Cover}\xspace}
\newcommand{\labelcover}{\textsc{Label Cover}\xspace}
\newcommand{\etal}{{\em et al.\ }}

\newcommand{\bn}{\bits^n}

\newcommand{\eps}{\epsilon}

\newcommand{\sgn}{\mathrm{pos}}

\newcommand{\cL}{{\cal L}}
\newcommand{\calE}{{\cal{E}}}

\newcommand{\calP}{{\cal P}}

\newcommand{\calX}{{\cal X}}
\newcommand{\calY}{{\cal Y}}
\newcommand{\calQ}{{\cal Q}}

\newcommand{\calR}{{\cal R}}
\newcommand{\NP}{\mathrm{NP}}

\newcommand{\x}{{\boldsymbol{x}}}

\newcommand{\y}{{\boldsymbol{y}}}

\newcommand{\z}{{\boldsymbol{z}}}

\newcommand{\w}{{\boldsymbol{w}}}

\newcommand{\mcl}[1]{\mathcal {#1}}

\newcommand{\cD}{\mcl{D}}
\newcommand{\cP}{\mcl{P}}
\renewcommand{\vec}[1]{{\bm{#1}}}

\newcommand{\iprod}[1]{\langle #1\rangle}

\newcommand{\bits}{\{0,1\}}

\newcommand{\vnote}[1]{}
\newcommand{\pnote}[1]{}
\newcommand{\yinote}[1]{}
\newcommand{\vinote}[1]{}

    \newcommand{\truncate}{\mathsf{Truncate}}
    \newcommand{\bvec}[2]{\vec{#1}^{\{#2\}}}
    \newcommand{\bsup}[2]{#1^{\{#2\}}}

    \newcommand{\cA}{\mathcal{A}}
        \newcommand{\cB}{\mathcal{B}}

\title{{\bf Agnostic Learning of Monomials by Halfspaces is Hard}\thanks{An extended abstract appeared in the Proceedings of the 50th IEEE Symposium on Foundations of Computer Science, 2009.}}

\author{{Vitaly Feldman}\thanks{IBM Almaden Research Center, San Jose, CA. {\tt vitaly@post.harvard.edu}.} \and {Venkatesan Guruswami}\thanks{Computer Science Department, Carnegie Mellon University, Pittsburgh, PA. {\tt guruswami@cmu.edu}.} \and {Prasad
  Raghavendra}\thanks{College of Computing, Georgia Institute of Technology, Atlanta, GA. {\tt praghave@cc.gatech.edu.} Some of this work was done when visiting Carnegie Mellon University.} \and {Yi Wu}\thanks{IBM Almaden Research Center, San Jose, CA. {\tt wuyi@us.ibm.com}. Most of this work was done when the author was at Carnegie Mellon University.}}

\begin{document}

\maketitle
\thispagestyle{empty}
\begin{abstract}
  We prove the following strong hardness result for learning: Given a
  distribution of labeled examples from the hypercube such that there
  exists a {\em monomial} consistent with $(1-\eps)$ of the examples,
  it is $\NP$-hard to find a {\em halfspace} that is correct on
  $(1/2+\eps)$ of the examples, for arbitrary constants $\eps
  > 0$. In learning theory terms, weak agnostic learning of monomials
  is hard, even if one is allowed to output a hypothesis from the much
  bigger concept class of halfspaces. This hardness result  subsumes a long line of  previous results, including two  recent   hardness results for the proper learning of monomials and halfspaces. As an immediate corollary of our result we show that weak agnostic learning of {\em decision lists} is  $\NP$-hard.

  Our techniques are quite different from previous hardness proofs for
  learning. We define distributions on positive and negative
  examples for monomials whose first few moments match. We use the
  {\em invariance principle} to argue that {\em regular}  halfspaces (all of whose
  coefficients have small absolute value relative to the total
  $\ell_2$ norm) cannot distinguish between distributions whose first
  few moments match. For highly non-regular subspaces, we use a
  structural lemma from recent work on fooling halfspaces to argue
  that they are ``junta-like'' and one can zero out all but the top
  few coefficients without affecting the performance of the
  halfspace. The top few coefficients form the natural list decoding
  of a halfspace in the context of dictatorship tests/Label Cover
  reductions.

 We note that unlike previous invariance principle based proofs which
  are only known to give Unique-Games hardness, we are able to reduce
  from a version of Label Cover problem that is known to be NP-hard.
 This has inspired follow-up work on bypassing the Unique Games conjecture in some optimal geometric inapproximability results.

 \end{abstract}

\newpage
\section{Introduction}

Boolean conjunctions (or {\em monomials}), decision lists, and halfspaces are among the
most basic concept classes in learning theory.  They are all
long-known to be efficiently PAC learnable, when the given examples are guaranteed to be consistent with a function from any of these concept classes~\cite{valiant,BEHW-occam,Rivest:87}. However,
in practice data is often noisy or too complex to be consistently explained by a simple concept. A common practical approach to such problems is to find a predictor in a certain space of hypotheses that best fits the given examples.
A general model for learning that addresses this scenario is the {\em  agnostic} learning model ~\cite{haussler,KSS94}. An {\it agnostic} learning algorithm for a class of functions $\mathcal{C}$ using a hypothesis space $\mathcal{H}$ is required to perform the following task:  Given examples drawn from some unknown distribution, the algorithm must find a hypothesis in $\mathcal{H}$ that classifies the examples nearly as well as is possible by a hypothesis from $\mathcal{C}$.   The algorithm is said to be a {\it proper} learning algorithm if $\mathcal{C}=\mathcal{H}$.


In this work we address the complexity of agnostic learning of monomials by algorithms that output a halfspace as a hypothesis. Learning methods that output a halfspace as a hypothesis such as Perceptron \cite{Rosenblatt:58}, Winnow \cite{littlestone}, Support Vector Machines \cite{Vapnik:98} as well as most boosting algorithms are well-studied in theory and widely used in practical prediction systems. These classifiers are often applied to labeled data sets which are not linearly separable. Hence it is of great interest to determine the classes of problems that can be solved by such methods in the agnostic setting. In this work we demonstrate a strong negative result on agnostic learning by halfspaces. We prove that non-trivial agnostic learning of even the relatively simple class of monomials by halfspaces is an NP-hard problem.

\begin{theorem}\label{thm:main}
For any constant $\eps > 0$, it is $\NP$-hard
to find a halfspace that correctly labels $(1/2+\eps)$-fraction of
given examples over $\bn$ even when there exists a
monomial that agrees with a $(1-\eps)$-fraction of the
examples.
\end{theorem}
Note that this hardness result is essentially optimal since
it is trivial to find a hypothesis with agreement rate $1/2$
--- output either the function that is always $0$ or the function that
is always $1$. Also note that  Theorem \ref{thm:main} measures agreement of a halfspace and a monomial with the given set of examples rather than the probability of agreement of $h$ with an example drawn randomly from an unknown distribution. Uniform convergence results based on the VC dimension imply that these settings are essentially equivalent (see for example \cite{haussler,KSS94}).

The class of monomials is a subset of the class of decision lists which in turn is a subset of the class of halfspaces.  Therefore our result immediately implies an optimal hardness result for proper agnostic learning of decision lists.

\subsection*{Previous work}

Before describing the details of the prior body of work on hardness results for learning, we note that our result {\em subsumes all these results} with just one exception (the hardness of learning monomials by $t$-CNFs~\cite{KS-dnf}). This is because we obtain the optimal inapproximability factor
{\em and} allow learning of monomials by the much richer class of halfspaces.

The results of the paper are noteworthy in the broader context of hardness of approximation.  Previously, hardness proofs based on the invariance principle were only known to give Unique-Games hardness.  In this work, we are able to harness invariance principles to show NP-hardness result by working with a version of Label Cover whose projection functions are only required to be {\em unique-on-average}. This could be one potential approach to revisit the many strong inapproximability results conditioned on the Unique Games conjecture (UGC), with an eye towards bypassing the UGC assumption. Such a goal was achieved for some geometric problems recently~\cite{GRSW10}; see Section \ref{sec:bypass-ugc}.

Agnostic learning of monomials, decision lists and halfspaces has been studied in a number of previous works. Proper agnostic learning of a class of functions $\mathcal{C}$ is equivalent to the ability to come up with a function in  $\mathcal{C}$ which has the optimal agreement rate with the given set of examples and is also referred to as the {\em Maximum
Agreement} problem for a class of function $\mathcal{C}$.

The Maximum Agreement problem for halfspaces is equivalent to the so-called Hemisphere problem and is long known to be NP-complete \cite{JohnsonPreparata:78,GareyJohnson:79}.
Amaldi and Kann \cite{amaldi-kann} showed that Maximum Agreement for halfspaces is NP-hard to approximate within $\frac{261}{262}$ factor. This was later improved by Ben-David \etal\cite{BenDavidEL:03}, and Bshouty and Burroughs \cite{BshoutyBurroughs:06} to approximation factors $\frac{415}{418}$, and $\frac{84}{85}$, respectively.
An optimal inapproximability result was established independently by Guruswami and Raghavendra~\cite{GR06J}
and Feldman \etal ~\cite{FGKP06J} showing NP-hardness of approximating the Maximum Agreement problem for halfspaces within
$(1/2+\eps)$ for every constant $\eps > 0$. The reduction in \cite{FGKP06J} requires examples with real-valued coordinates, whereas the proof in \cite{GR06J} also works for examples drawn from the Boolean hypercube.

The Maximum Agreement problem for monotone monomials was shown to be
NP-hard by Angluin and Laird~\cite{angluin-laird}, and NP-hardness for general monomials was shown by Kearns and Li~\cite{kearns-li}. The hardness of approximating the
maximum agreement within $\frac{767}{770}$ was shown by Ben-David \etal \cite{BenDavidEL:03}. The factor was subsequently improved to $58/59$ by Bshouty and Burroughs~\cite{BshoutyBurroughs:06}. Finally, Feldman \etal ~\cite{feldman,FGKP06J} showed a tight inapproximability
result, namely that it is
NP-hard to distinguish between the instances where $(1-\eps)$-fraction of the
labeled examples are consistent with some monomial and instances
where every monomial is consistent with at most $(1/2+\eps)$-fraction of the
examples. Recently, Khot and Saket~\cite{KS-dnf} proved a similar
hardness result even when a $t$-CNF is allowed as output hypothesis
for an arbitrary constant $t$ (a $t$-CNF is the conjunction of several
clauses, each of which has at most $t$ literals; a monomial is thus a
$1$-CNF).

For the concept class of decisions lists, APX-hardness (or hardness to approximate within some constant factor) of the Maximum Agreement problem was shown by Bshouty and Burroughs~\cite{BshoutyBurroughs:06}.
As mentioned above, our result subsumes all these results with the exception of \cite{KS-dnf}.

A number of hardness of approximation results are also known for the complementary problem of minimizing disagreement for each of the above concept classes~\cite{KSS94,HoffgenHS:95,ABSS,BshoutyBurroughs:02min,feldman,FGKP06J}. Another well-known evidence of the hardness of agnostic learning of monomials is that even a non-proper agnostic learning of monomials would give an algorithm for learning DNF --- a major open problem in learning theory~\cite{LeeBW:95}. Further, Kalai \etal proved that even agnostic learning of halfspaces with respect to the uniform distribution implies learning of parities with random classification noise --- a long-standing open problem in learning theory and coding \cite{KalaiKMS:08}.

Monomials, decision lists and halfspaces are known to be efficiently learnable in the presence of more benign {\em random} classification noise ~\cite{angluin-laird,Kearns:98,kearns-schapire,Bylander:94,BFKV,cohen}. Simple online algorithms like Perceptron and Winnow learn halfspaces when the examples can be separated with a significant {\em margin} (as is the case if the examples are consistent with a monomial) and are known to be robust to a very mild amount of adversarial noise \cite{Galant:90,AuerWarmuth:98,GentileWarmuth:98}. Our result implies that these positive results will not hold when the adversarial noise rate is $\eps$ for any constant $\eps >0$.

Kalai \etal gave the first non-trivial algorithm for agnostic learning monomials in time $2^{\tilde{O}(\sqrt{n})}$ \cite{KalaiKMS:08}. They also gave a breakthrough result for agnostic learning of halfspaces with respect to the uniform distribution on the hypercube up to any constant accuracy (and analogous results for a number of other settings). Their algorithms output linear thresholds of parities as hypotheses. In contrast, our hardness result is for algorithms that output a halfspace (which is a linear threshold of single variables).
\ignore{An agnostic algorithm for learning decision trees under the uniform
distribution, albeit using membership queries, was recently discovered
in \cite{GKK}.}
\paragraph{Organization of the paper:} We sketch the idea of our proof in Section~\ref{sec:overview}. We define some  probability and analytical  tools in Section~\ref{sec:tools}. In Section~\ref{sec:test} we define the \textit{dictatorship test}, which is an important gadget for the  hardness reduction. For the purpose of illustration, we also show why this dictatorship test already suffices  to prove Theorem~\ref{thm:main} assuming the Unique Games Conjecture~\cite{Khot-UGC}. In Section 5, we describe a reduction from a variant of the  \labelcover problem to prove Theorem~\ref{thm:main} under the assumption that $\P\ne \NP$.

\paragraph{Notation:}
We use $0$ to encode ``False'' and $1$ to encode ``True''. We denote $\sgn(t) : \R\to \{0,1\}$ as the  indicator function of whether $t\geq 0$; i.e., $\sgn(t) = 1$ when $t\geq 0$ and $\sgn(t)= 0$ when $t< 0$.

For $\vec{x} =(x_1,x_2,\dots,x_n) \in \bits^n$, $\vec{w}\in\R^n$, and $\theta\in \R$, a halfspace $h(\vec{x})$ is a Boolean function of the form $\sgn(\vec{w}\cdot \vec{x}-\theta)$; a monomial  (conjunction) is a function of the form $\bigwedge_{i\in S} s_i$, where $S\subseteq [n]$ and $s_i$ is the literal of $x_i$ which can represent either  $x_i$ or $\neg x_i$; a disjunction is a function of the form  $\bigvee_{i\in S} s_i$. One special case of monomials is the  function $f(x) = x_i$ for some $i\in [n]$, also referred to as the $i$-th {\it dictator} function.

\section{Proof Overview}\label{sec:overview}
We prove Theorem \ref{thm:main} by exhibiting a reduction from
the \klabelcover problem, which is a particular variant of the \labelcover
problem.  The \klabelcover problem is defined as follows:
\begin{definition}\label{def:klabelcover}
        For positive integer $M,N$ that $M\geq N$ and $k\geq 2$, an instance of \klabelcover
        $\cL(G(V,E),M,N,\{\pi^{v,e}|e\in E, v\in e\})$ consists of a
        $k$-uniform connected (multi-)hypergraph $G(V,E)$ with vertex set $V$ and
        an edge multiset $E$; a set of functions $\{\pi^{v_i,e}\}_{i=1}^{k}$.
        Every hyperedge $e = (v_1,\ldots,v_{k})$ is associated with a $k$-tuple
        of projection functions $\{\pi^{v_i,e}\}_{i=1}^{k}$ where
        $\pi^{v_i,e}: [M] \to [N]$.

        A vertex labeling $\Lambda$ is an assignment of labels to vertices
        $\Lambda : V\to [M]$. A
        labeling $\Lambda$ is said to {\em strongly} satisfy an edge $e$ if
        $\pi^{v_i,e}(\Lambda(v_i)) = \pi^{v_j,e}(\Lambda(v_j)))$ for every
        $v_i,v_j\in e$. A labeling $\Lambda$ {\em weakly} satisfies edge $e$ if
        $\pi^{v_i,e}(\Lambda(v_i)) = \pi^{v_j,e}(\Lambda(v_j)))$ for some $v_i,
        v_j\in e$, $v_i \neq v_j$.
\end{definition}
The goal in \labelcover is to find a vertex labeling that satisfies as many
edges (projection constraints) as possible.

\subsection{Hardness assuming the Unique Games conjecture}

For the sake of clarity, we first sketch the proof of
Theorem \ref{thm:main} with a reduction from the  \kuniquelabelcover
problem which is a special case of \klabelcover where $M= N$ and all the projection functions $\{\pi^{v,e}| v \in e, e \in E\}$ are bijections.  The following inapproximability result  ~\cite{KhotR08} for
\kuniquelabelcover is  equivalent to the Unique
Games Conjecture of Khot \cite{Khot-UGC}.
\begin{conjecture}\label{ugconj}
        For every constant $\eta > 0$ and a positive integer $k$, there exists an integer
        $R_0$ such that for all positive integers $R > R_0$, given an instance $\cL(G(V,E),R,R,\{\pi^{v,e} | e\in E, v\in e\})$ it is
        $\NP$-hard to distinguish between,
        \begin{itemize}\itemsep=0ex
                \item strongly satisfiable instances: there exists a
                        labeling $\Lambda : V \to [R]$ that {\it strongly
                        satisfies} $1-k\eta$ fraction of the edges $E$.
                \item almost unsatisfiable instances: there is no labeling that weakly
                        satisfies $\frac{2k^2}{R^{\eta/4}}$ fraction of the edges.
        \end{itemize}
\end{conjecture}

Given an instance $\cL$ of \kuniquelabelcover, we will produce a distribution $\cD$ over labeled examples such that the following holds: if $\cL$ is a
strongly satisfiable instance,
then there is a disjunction that agrees with the label on  a randomly chosen example  with probability at least $1-\eps$,  while if $\cL$ is an almost unsatisfiable
instance then no halfspace agrees with the label on  a random example from $\cD$ with probability more than $\frac{1}{2}+\eps$. Clearly, such a reduction
implies Theorem \ref{thm:main} assuming the Unique Games
Conjecture but with disjunctions in place of conjunctions. De Morgan's law and the fact that a negation of a halfspace is a halfspace then imply that the statement is also true for monomials (we use disjunctions only for convenience).

Let $\cL$ be an instance of \kuniquelabelcover on
hypergraph $G = (V,E)$ and a set of labels $[R]$.  The examples we
generate will have $|V| \times R$ coordinates, i.e., belong to
$\{0,1\}^{|V| \times R}$.  These coordinates are to be thought of as
one block of $R$ coordinates for every vertex $v \in V$.  We will
index the coordinates of $\vec{x} \in \{0,1\}^{|V| \times R}$ as
$\vec{x} = (x_{v}^{(\rr)})_{v \in V, \rr\in [R]}$.

For every labeling $\Lambda : V \to [R]$ of the instance, there is a
corresponding disjunction over $\{0,1\}^{|V| \times R}$
given by,
$$h(\vec{x}) = \bigvee_{v} x_{v}^{(\Lambda(v))}.$$
Thus, using a label $\rr$ for a vertex $v$ is encoded as including
the literal $x_v^{(\rr)}$ in the disjunction.
Notice that an arbitrary halfspace over $\{0,1\}^{|V| \times
R}$ need not correspond to any labeling at all.  The idea
would be to construct a distribution on examples which ensures that any halfspace agreeing with at least $\frac{1}{2} + \eps$ fraction of random examples somehow corresponds to a labeling of
$\Lambda$ weakly satisfying a constant fraction of the edges in $\cL$.

Fix an edge $e = (v_1,\ldots,v_k)$.  For the sake of exposition, let
us assume $\pi^{v_i,e}$ is the identity permutation for every $i \in
[k]$.  The general case is not anymore complicated.

 For the edge $e$,
we will construct  a distribution on examples $\cD_e$ with the following properties:
\begin{itemize}\itemsep=0ex
        \item All coordinates $x_v^{(\rr)}$ for a vertex $v \notin e$
                are fixed to be zero.  Restricted to these
                examples, the halfspace $h$ can be written as
                $h(\vec{x}) = \sgn(\sum_{i \in [k]}
                \iprod{\vec{w}_{v_i},\vec{x}_{v_i}} - \theta)$.
        \item For any label $\rr \in [R]$, the labeling $\Lambda(v_1) =
                \ldots = \Lambda(v_k) = \rr$ {\em strongly} satisfies the
                edge $e$.  Hence, the corresponding disjunction
                $\vee_{i \in [k]} x_{v_i}^{(\rr)}$ needs to have agreement
                $\geq 1-\eps$ with the examples from $\cD_e$.
        \item There exists a decoding procedure that given a halfspace
                $h$ outputs a labeling $\Lambda_h$ for $\cL$ such that,
                if $h$ has agreement
                $\geq \frac{1}{2}+\eps$ with the examples from $\cD_e$,
                then $\Lambda_h$ {\it weakly} satisfies the edge $e$ with
                non-negligible probability.
\end{itemize}

For conceptual clarity, let us rephrase the above requirement as a
testing problem. Given a halfspace $h$, consider a randomized
procedure that samples an example $(\vec{x},b)$ from the
distribution $\cD_e$, and accepts if $h(\vec{x}) = b$.  This amounts to
a test that checks if the function $h$ corresponds to a consistent
labeling.  Further, let us suppose the halfspace $h$ is given by
$h(\vec{x}) = \sgn\left(\sum_{v \in V} \iprod{\vec{w}_v,\vec{x}_v} -
\theta\right).$
Define the linear function $f_v : \{0,1\}^{R} \to \R$ as $f_v(\vec{x}_v)
= \iprod{\vec{w}_v,\vec{x}_v}$.  Then, we have $h(\vec{x}) = \sgn(\sum_{v \in
V} f_v(\vec{x}_v) -\theta)$.

For a halfspace $h$ corresponding to a labeling $\Lambda$,  we will have
$f_v(\vec{x}_v) = x_v^{(\Lambda(v))}$ -- a dictator function.   Thus, in the intended solution every linear
function $f_v$ associated with the halfspace $h$ is a dictator function.

Now, let us again restate the above testing problem in terms of these
linear functions.  For succinctness, we write $f_{i}$ for the linear
function $f_{v_i}$.  We need a randomized procedure that does the
following:
\begin{quote}
Given $k$ linear functions $f_{1},\ldots,f_{k} : \{0,1\}^R \to
\R$, queries the functions at one point each (say $\vec{x}_{1},\ldots,\vec{x}_{k}$
respectively), and accepts if $\sgn(\sum_{i=1}^k f_{i}(\vec{x}_{i}) -
\theta) = b$.
\end{quote}
The procedure must satisfy,
\begin{itemize}
        \item (Completeness) If each of the linear functions $f_i$ is the $\rr$'th
                dictator function for some $\rr \in [R]$, then the
                test accepts with probability $1-\eps$.
        \item (Soundness)  If the test accepts with probability
                $\frac{1}{2}+\eps$, then at least {\em two} of the linear
                functions  are {\em close} to the same dictator function.
\end{itemize}
A testing problem of the above nature is referred to as a
{\em Dictatorship Testing} and is a recurring theme in hardness of
approximation.

Notice that the notion of a linear function being {\it close} to a
dictator function is not formally defined yet.  In most applications,
a function is said to be close to a dictator if it has {\it
influential} coordinates.  It is easy to see that this notion is not
sufficient by itself here.  For example, in the linear function
$\sgn(10^{100}x_1 +
x_2 - 0.5)$,  although the coordinate $x_2$ has little influence on
the linear function, it has significant influence on the
halfspace.

We resolve this problem by using the notion of {\it critical index}
(Definition \ref{def:cti}) that was introduced
in \cite{Ser07} and has found numerous applications in the analysis of
halfspaces \cite{MORS09,OS08,DGJSV09}.  Roughly speaking, given a linear
function $f$, the idea is to recursively delete its influential coordinates until
there are none left.  The total number of coordinates so deleted is
referred to as the critical index of $f$.  Let $c_{\tau}(\vec{w}_i)$ denote the critical index of $\vec{w}_i$, and
let $C_{\tau}(\vec{w}_i)$ denote the set of $c_{\tau}(\vec{w}_i)$
largest coordinates of $\vec{w}_i$.  The linear function $l$ is said
to be {\it close} to the $i$'th dictator function for every $i$ in
$C_{\tau}(\vec{w}_i)$.  A function is {\it far}
from every dictator if it has critical index $0$ -- no influential coordinate to delete.

An important issue is that the critical index of a linear function can be much larger than the number of influential coordinates and cannot be appropriately bounded.
 In other words, a linear function can be close to a
large number of dictator functions, as per the definition above.  To
counter this, we employ a structural lemma
about halfspaces that was used in the recent work on fooling
halfspaces with limited independence~\cite{DGJSV09}. Using this lemma, we are
able to prove that if the critical index is large, then one can in fact
zero out the coordinates of $\vec{w}_i$ outside the $t$ largest
coordinates for some large enough $t$, and the agreement of the halfspace $h$ only changes by a negligible amount!
Thus, we first carry out the zeroing operation for
all linear functions with large critical index.

We now describe the above construction and analysis of the
dictatorship test in some more detail.  It is convenient to think of
the $k$ queries $\vec{x}_1,\ldots,\vec{x}_k$ as the rows of a $k \times
R$ matrix with $\{0,1\}$ entries.  Henceforth, we will refer to
matrices $\{0,1\}^{k \times R}$ and their rows and columns.

We construct two distributions
$\cD_0,\cD_1$ on $\{0,1\}^k$ such that for $s \in \{0,1\}$, we have $\Pr_{x
  \in \cD_s} \bigl[ \vee_{i=1}^k x_i = s \bigr] \ge 1-\eps/2$ for
$\eps=o_k(1)$ (this will ensure the completeness of the reduction,
i.e., certain disjunctions pass with high probability). Further, the
distributions $\cD_0, \cD_1$ will be carefully chosen to have matching first four
moments. This will be used in the soundness analysis where we will use
an {\em invariance principle} to infer structural properties of
halfspaces that pass the test with probability noticeably greater than
$1/2$.

We define the distribution $\tilde{\cD}_s^R$ on matrices $\{0,1\}^{k \times R}$
by sampling $R$ columns independently according to $\cD_s$, and then
perturbing each bit with a small probability $\eps/2$.  We define
the following test (or equivalently, distribution on examples): given a
halfspace $h$ on $\{0,1\}^{k \times R}$, with probability $1/2$ we
check $h(\vec{x}) = 0$ for a sample $\vec{x} \in \tilde{\cD}_0^R$, and with
probability $1/2$ we check $h(\vec{x})=1$ for a sample $\vec{x} \in
\tilde{\cD}_1^R$.

\smallskip
\noindent\textbf{Completeness:} By construction, each of the $R$ disjunctions $\mathsf{OR}_j(\vec{x})
= \vee_{i=1}^k x^{(j)}_i$ passes the test with probability at least
$1-\eps$ (here $x^{(j)}_i$ denotes the entry in the $i$'th row and
$j$'th column of $\vec{x}$).

\smallskip
\noindent\textbf{Soundness:}
For the soundness analysis, suppose $h(\vec{x}) = \sgn(\langle
\vec{w},\vec{x}\rangle-\theta)$  is a halfspace that passes the
test with probability at least $1/2+\epsilon$.  The halfspace $h$ can be written
in two ways by expanding the inner product $\iprod{\vec{w},\vec{x}}$ along
rows and columns, i.e.,
$h(\vec{x}) = \sgn(\sum_{i=1}^k \langle
\vec{w}_i,\vec{x}_i\rangle - \theta)= \sgn(\sum_{i=1}^R \langle
\vec{w}^{(i)},\vec{x}^{(i)}\rangle - \theta) .$
Let us denote $f_i(\vec{x}) =
\iprod{\vec{w}_i,\vec{x}_i}$.

First, let us see why the linear functions
$\iprod{\vec{w}_i,\vec{x}_i}$ must be close to {\em some} dictator.
Note that we need to show that two of the linear functions are close
to the {\it same} dictator.

Suppose each of the linear functions $f_{i}$ is not {\it close} to any
dictator.  In other words, for each $i$, no single coordinate of the
vector $\vec{w}_i$ is too large (contains more than $\tau$-fraction of
the $\ell_2$ mass $\|\vec{w}_i\|_2$ of vector $\vec{w}_i$ ).  Clearly,
this implies that no single column of the matrix $\vec{w}$ is too {\it
large}.

Recall that the halfspace is given by $ h(\vec{x}) = \sgn(\sum_{j
\in [R]} \iprod{\vec{w}^{(j)},\vec{x}^{(j)}}-\theta).$
Here $l(\vec{x}) = \sum_{j \in [R]} \iprod{\vec{w}^{(j)},\vec{x}^{(j)}}-\theta$ is a
degree $1$ polynomial into which we are substituting values from two
product distributions $\cD_0^R$ and $\cD_1^R$.  Further, the
distributions $\cD_0$ and $\cD_1$ have matching moments up to order
$4$ by design.  Using the invariance principle, the distribution of
$l(\vec{x})$ is roughly the same, whether $\vec{x}$ is from $\cD_0^R$ or
$\cD_1^R$.  Thus, by the invariance principle,  the halfspace
$h$ is unable to distinguish between the distributions $\cD_0^R$
and $\cD_1^R$ with a noticeable advantage.

Further, suppose no two linear functions $f_i$ are {\it close} to the same
dictator, i.e., $C_{\tau}(\vec{w}_i) \cap C_{\tau}(\vec{w}_j) =
\emptyset$.
In this case, we condition on the
values of $x^{(j)}_i$ for $j\in C_{\tau}(\vec{w}_i)$.  Since $C_{\tau}(\vec{w}_i) \cap C_{\tau}(\vec{w}_j) =
\emptyset$, this conditions at
most \emph{one} value in each column.  Therefore, the conditional distribution on each
column in cases $\cD_0$ and $\cD_1$ still have matching first three moments.  We thus apply the
invariance principle using the fact that after deleting the coordinates in
$C_{\tau}(\vec{w}_i)$, all the remaining
coefficients of the weight vector $\vec{w}$ are small (by definition of
critical index).  This implies that $C_{\tau}(\vec{w}_i) \cap
C_{\tau}(\vec{w}_j) \neq \emptyset$ for some two rows $i,j$ and finishes
the proof of the soundness claim.

The above consistency-enforcing test almost immediately yields the
Unique Games hardness of weak learning disjunctions by halfspaces
via standard methods.

\vspace{-1ex}
\subsection{Extending to NP-hardness}
\vspace{-1ex}

To prove NP-hardness as opposed to hardness assuming the Unique Games conjecture, we reduce a version of Label Cover to our problem. This requires a more complicated consistency check, and we have to overcome
several additional technical obstacles in the proof.

The main obstacle encountered in transferring the
dictatorship test to a Label Cover-based hardness is one that
commonly arises for several other problems.  Specifically, the projection
constraint on an edge $e = (u,v)$ maps a large set of labels
$ \calR=\{\rr_1,\ldots,\rr_d\}$ corresponding to a vertex $u$ to a single
label $\rr$ for the vertex $v$.  While composing the Label Cover
constraint $(u,v)$ with the dictatorship test, all labels in $\calR$ have
to be necessarily {\it equivalent}.  In several settings including
this work, this requires the coordinates corresponding to labels in
$\calR$ to be {mostly identical}!  However, on making the coordinates
corresponding to $\calR$ identical, the prover corresponding to $u$ can
determine the identity of edge $(u,v)$, thus completely destroying the
soundness of the composition.  In fact, the natural extension of
the Unique Games-based reduction for {\sc MaxCut} \cite{KhotKMO07} to
a corresponding Label Cover hardness fails primarily for this reason.

Unlike {\sc MaxCut} or other Unique Games-based reductions, in our
case, the soundness of the dictatorship test is required to hold
against a specific class of functions, i.e, halfspaces.  Harnessing
this fact, we execute the reduction starting from a Label Cover
instance whose projections are {\it unique on average}.  More
precisely, a {\it smooth} Label Cover (introduced in \cite{Khot03}) is one in which for every vertex $u$, and a pair of labels $\rr,\rr'$, the labels $\{\rr,\rr'\}$ project
to the same label with a tiny probability over the choice of the edge
$e = (u,v)$.  Technically, we express the error term in the invariance
principle as a certain fourth moment of the coefficients of the halfspace, and use the smoothness to bound this error term for most edges of the Label Cover instance.

\vspace{-1ex}
\subsection{Bypassing the Unique Games conjecture}
\label{sec:bypass-ugc}

Unlike previous invariance principle based proofs which are only known to give Unique-Games hardness, we are able to reduce from a version of the Label Cover problem, based on \emph{unique on average} projections, that can be shown to be NP-hard.  It is of great interest to find other applications where a {\it weak uniqueness} property like the smoothness condition mentioned above can be used to convert a Unique-Games hardness result to an unconditional NP-hardness
result. Indeed, inspired by the success of this work in avoiding the UGC assumption and using some of our methods, follow-up work has managed to bypass the Unique Games conjecture in some optimal geometric inapproximability results~\cite{GRSW10}. To the best of our knowledge, the results of \cite{GRSW10} are the first NP-hardness proofs showing a tight inapproximability factor that is related to fundamental parameters of Gaussian space, and among the small handful of results where optimality of a non-trivial semidefinite programming based algorithm is shown under the assumption ${\rm P} \neq {\rm NP}$. We hope that this paper has thus opened the avenue to convert at least some of the many tight Unique-Games hardness results to NP-hardness results.

\section{Preliminaries}\label{sec:tools}

In this section, we define two  important tools in our analysis: i) critical index, ii) invariance principle.

\subsection{Critical Index}
The notion of critical index was first introduced by Servedio \cite{Ser07} and plays an important role in the analysis of  halfspaces in \cite{MORS09,OS08,DGJSV09}.

\begin{definition}\label{def:cti}
        Given any real vector $\vec{w} =(w^{(1)},w^{(2)},\ldots,w^{(n)})\in
        \R^n$.
        Reorder the coordinates by decreasing absolute value, i.e.,
         $|w^{(i_1)}|\geq |w^{(i_2)}|\geq \ldots \geq|w^{(i_n)}|$ and denote $\sigma_t^2 = \sum_{j=t}^{n}|w^{(i_j)}|^2$. For $0\leq
         \tau\leq 1$,  the $\tau$-critical index of the vector
         $\vec{w}$ is
        defined to be the smallest index $k$ such  $|w^{(i_k)}|\leq \tau
        \sigma_k$. If no such $k$ exists ($\forall k$, $|w^{(i_k)}| >
        \tau \sigma_k$),  the $\tau$-critical index is defined to be
        $+\infty$. The vector $\vec{w}$ is said to be $\tau$-regular if the $\tau$-critical index is $1$.
\end{definition}
A simple observation from \cite{DGJSV09} is that if the critical index of a sequence is large then the sequence must contain a geometrically decreasing subsequence.
\begin{lemma}(Lemma $5.5$ in \cite{DGJSV09}) \label{lem:st} Given a vector
        $\vec{w} = (w^{(i)})_{i=1}^n$ such that $|w^{(1)}|\geq |w^{(2)}|\geq \ldots \geq
        |w^{(n)}|$, if the $\tau$-critical index of the vector
        $\vec{w}$ is larger than $l$, then for any $1\leq i\leq j\leq l+1$,
        \[|w^{(j)}|\leq \sigma_j\leq (\sqrt{1-\tau^2})^{j-i}\sigma_i
        \leq (\sqrt{1-\tau^2})^{j-i} |w^{(i)}|/\tau.\] In particular,
        if $j > i + (4/\tau^2) \ln(1/\tau) $ then $|w^{(j)}|\leq
        |w^{(i)}|/3$.
\end{lemma}

For a $\tau$-regular weight vector, the following lemma bounds the probability that its weighted sum falls into a small interval under certain distributions on the points.
The proof is in Appendix \ref{sec:lem23}.
\begin{lemma}\label{lem:spread}
        Let $\vec{w} \in \R^n$ be a $\tau$-regular vector $\vec{w}$,   and $\sum
        |w^{(i)}|^2 = 1$.
        $\cD$ is a distribution over $\bits^n$.  Define a distribution
        $\tilde{\cD}$ on $\bits^n$ as follows:  to generate $\vec{y}$
        from $\tilde{\cD}$, first sample $\vec{x}$ from
        $\cD$ and then define,
        \begin{equation*}
                y^{(i)} = \begin{cases}  x^{(i)}&  \text{ with
                                probability }  1- \gamma \\
                                \text{random bit} & \text{ with
                                probability } \gamma.
                \end{cases}
                \end{equation*}
        Then for any interval $[a,b]$, we have
        \begin{eqnarray*}
                \Pr\Big[ \iprod{\vec{w},\vec{y}} \in [a,b] \Big] \leq
                \frac{4|b-a|}{\sqrt{\gamma}} + \frac{4\tau}{\sqrt{\gamma}} +2e^{-\frac{\gamma^2}{2\tau^2}}.
        \end{eqnarray*}
\end{lemma}
Intuitively, by the Berry-Esseen Theorem, $\iprod{\vec{w},\vec{y}}$ is $\tau$ close to the Gaussian distribution if each $y^{(i)}$ is a random bit; therefore we can bound the probability that $\iprod{\vec{w},\vec{y}}$  falls into the interval $[a,b]$. In above lemma, each $y^{(i)}$ has probability $\gamma$ to be a random bit, then $\gamma$ fraction of $y^{(i)}$ is set to be a random bit and we can similarly bound the probability that $\iprod{\vec{w},\vec{y}}$ falls into the interval $[a,b]$.

\begin{definition}\label{def:top}
For a vector $\vec{w} \in \R^n$, define set of indices
        $\bb_{t}(\vec{w}) \subseteq [n]$ as the set of indices containing the $t$
        biggest coordinates of $\vec{w}$ by absolute value.
Suppose its $\tau$-critical index is $c_{\tau}$, define set of indices
        $C_{\tau}(\vec{w}) = \bb_{c_{\tau}}(\vec{w})$. In other words, $C_{\tau}(\vec{w})$ is the set of indices whose
        deletion makes the vector $w$ to be $\tau$-regular.   \end{definition}

\begin{definition}
        For a vector $\vec{w} \in \R^n$ and a subset of indices $S
        \subseteq [n]$, define the vector $\truncate(\vec{w},S) \in \R^n$ as:
        \begin{eqnarray*}
                (\truncate(\vec{w},S))^{(i)} = \begin{cases} w^{(i)} \text{ if } i \in
                        S \\ 0 \text{ otherwise }\end{cases}
        \end{eqnarray*}
\end{definition}

 As suggested by Lemma \ref{lem:st}, a weight vector with a large critical
 index has a geometrically decreasing subsequence. The following two
 lemmas use this fact to bound the probability that the weighted sum of a
 geometrically decreasing sequence of weights falls into a small interval.  First, we
 restate Claim 5.7 from \cite{DGJSV09} here.
 \begin{lemma}\label{lem:gm}[Claim 5.7, \cite{DGJSV09}]
       Let $\vec{w} = (w^{(1)},\ldots,w^{(T)})$ be such that $|w^{(1)}| \geq |w^{(2)}| \ldots \geq |w^{(T)}| \ge0$ and
        $|w^{(i+1)}|\leq|\frac{w^{(i)}}{3}|$ for $1\leq i\leq T-1$ . Then for any
        interval $ I = [\alpha-\frac{w^{(T)}}{6},\alpha+\frac{w^{(T)}}{6}]$ of length
        $\frac{|w^{(T)}|}{3}$, there is at most one point  $\vec{x}\in \{0,1\}^T$
        such that  $\iprod{\vec{w},\vec{x}} \in I$.
\end{lemma}

\begin{lemma}\label{lem:geometricsmallball1}
      Let $\vec{w} = (w^{(1)},\ldots,w^{(T)})$ be such that $|w^{(1)}| \geq |w^{(2)}| \ldots \geq |w^{(T)}| \ge0$ and
        $|w^{(i+1)}|\leq|\frac{w^{(i)}}{3}|$ for $1\leq i\leq T-1$.   Let $\cD$ be a
distribution over $\bits^T$.  Define a distribution
        $\tilde{\cD}$ on $\bits^T$ as follows:  To generate $\vec{y}$
        from $\tilde{\cD}$, sample $\vec{x}$ from
        $\cD$ and set
        \begin{equation*}
                y^{(i)} = \begin{cases}  x^{(i)}&  \text{ with
                                probability }  1- \gamma \\
                                \text{random bit} & \text{ with
                                probability } \gamma.
                \end{cases}
                \end{equation*}
        Then for any $\theta \in \R$ we have
        \begin{eqnarray*}
                \Pr\Big[ \iprod{\vec{w},\vec{y}} \in
                [\theta-\frac{w^{(T)}}{6},\theta+\frac{w^{(T)}}{6}] \Big] \leq
                \left(1- \frac{\gamma}{2}\right)^T .
        \end{eqnarray*}
\end{lemma}
\begin{proof}
  By Lemma \ref{lem:gm}, we know that for the interval $J
  =\left[\theta - \frac{|w^{T}|}{6}, \theta +
    \frac{|w^{T}|}{6}\right]$, there is at most one point $\vec{r} \in
  \bits^{T}$ such that $\iprod{\vec{w},\vec{r}}\in J$.  If no such
  $\vec{r}$ exists then clearly the probability is zero.  On the other
  hand, suppose there exists such an $\vec{r}$, then $\iprod{\vec{w},
    \vec{y}} \in J$ only if $(y_1^{(1)},y_1^{(2)},\ldots,y_1^{(T)}) =
  ({r}^{(1)},\ldots,{r}^{(T)})$ holds.

    Conditioned on any fixing of the bits $\vec{x}$, every bit
    $y^{(j)}$ is an independent random bit with probability
    $\gamma$.  Therefore, for every fixing of $\vec{x}$, for each $i \in [T]$, with probability at
    least  $\gamma/2$, $y^{(i)}$ is not equal to $r^{(i)}$. Therefore,
 $\Pr[ {y}^{(1)} = {r}^{(1)},{y}^{(2)} =
 {r}^{(2)},\ldots,{y}^{(T)} = {r}^{(T)} ]\leq
\left(1-\frac{\gamma}{2}\right)^T$.
\end{proof}

\subsection{Invariance Principle}

While invariance principles have been shown in various settings by
\cite{MosselOO05, Chatterjee05, Mossel08}, we restate a version of the principle well suited
for our application.  We present a self-contained proof for
it in Appendix~\ref{sec:invar}. 
\begin{definition}
A function $\Psi(x):\R\to \R$ for which fourth-order derivatives exist everywhere on $\R$ is said to be $K$-bounded if
$|\Psi''''(t)| \leq K$ for all $t\in \R$.
\end{definition}

\begin{definition}
Two ensembles of random variables $\calP = (p_{1},\ldots,p_{k})$ and
$\calQ=(q_1,\ldots,q_k)$ are said to have matching moments up to degree $d$
if for every multi-set $S$ of elements from $[k]$, $|S| \leq d$, we have
$\Ex[\prod_{i\in S}p_i] = \Ex[\prod_{i\in S}q_i]$.

\end{definition}

\begin{theorem} {(Invariance Principle)} \label{thm:invariance}
        Let $\cA = \{\bvec{A}{1},\ldots,\bvec{A}{R}\},
        \cB = \{\bvec{B}{1},\ldots,\bvec{B}{R}\}$ be families of ensembles of random
        variables with $\bvec{A}{i} =
        \{a^{(i)}_1,\ldots,a^{(i)}_{k_i}\}$ and $\bvec{B}{i} =
        \{b^{(i)}_1,\ldots,b^{(i)}_{k_i}\}$, satisfying the following
        properties:
        \begin{itemize}\itemsep=0ex
                \item For each $i \in [R]$, the random variables in
                        ensembles $(\bvec{A}{i},\bvec{B}{i})$ have
                        matching moments up to degree $3$.  Further
                        all the random variables in $\cA$ and $\cB$ are
                        bounded by $1$.
                \item The ensembles $\bvec{A}{i}$ are all independent
                        of each other, similarly the ensembles
                        $\bvec{B}{i}$ are independent of each other.
        \end{itemize}
        Given a set of vectors $\vec{l} = \{\bvec{l}{1},\ldots,\bvec{l}{R}\}
        (\bvec{l}{i} \in \R^{k_i})$, define the linear function
        $\vec{l} :
        \R^{k_1} \times \cdots \times \R^{k_R} \to \R$ as
        $$\vec{l}(\vec{x}) = \sum_{i\in[R]} \iprod{
                \bvec{l}{i}, \bvec{x}{i}} $$
        Then for a $K$-bounded function $\Psi : \R
        \to \R$ we have
        \begin{equation*} \left|  \Ex_{\cA}\left[\Psi\Big(\vec{l}(\cA) -\theta \Big)\right] -
                \Ex_{\cB}\left[\Psi\Big( \vec{l}(\cB)
                -\theta \Big)\right] \right| \leq K \sum_{i \in [R]}
                \|\bvec{l}{i}\|_1^4 \end{equation*}
        for all $\theta > 0$. Further, define the spread function
        $c(\alpha)$ corresponding to the ensembles $\cA,\cB$ and the linear function
        $\vec{l}$ as follows,
        \begin{align*}
        (\textbf{Spread Function: }&)  \text{For  } 1/2> \alpha >0
        \text{, let  }\\
        c(\alpha)  = \max\big(& \sup_{\theta} \Pr_{\cA}  \Big[ \vec{l}(\cA) \in [\theta - \alpha,
        \theta+\alpha] \Big],  \quad \sup_{\theta} \Pr_{\cB} \Big[
        \vec{l}(\cB) \in [\theta - \alpha,
        \theta+\alpha] \Big]\big)
        \end{align*}
        then for all $\theta$,
        \begin{eqnarray*}
                \left|
                \Ex_{\cA}\left[\sgn\left(\vec{l}(\cA)-\theta\right)\right] -
                 \Ex_{\cB}\left[\sgn\left( \vec{l}(\cB)
                -\theta\right)\right]
                  \right|  \leq
                   O\left(\frac{1}{\alpha^4} \right) \sum_{i \in [R]}
                \|\bvec{l}{i}\|_1^4 + 2c(\alpha).
        \end{eqnarray*}
\end{theorem}

Roughly speaking, the second part of the theorem states that  $\sgn$ function can be thought of  as $\frac{1}{\alpha^4}$-bounded with error parameter  $c(\alpha)$.

\section{Construction of the Dictatorship Test}\label{sec:test}

In this section we describe the construction of the dictatorship test which
will be the key ingredient in the
hardness reduction from \kuniquelabelcover.

\subsection{Distributions $\cD_0$ and $\cD_1$}
The dictatorship test is based on following two distributions $\cD_0$ and $\cD_1$
defined on $\bits^k$.
\begin{lemma}\label{lem:d01} For $k\in \N$, there exists two probability
        distributions $\cD_0$, $\cD_1$ on $ \bits^k$ such that for $x=(x_1,\ldots,x_k)$,
        $$\Pr_{\x\sim \cD_0}\{\text{every $x_l$ is $0$}\} \geq
        1-\frac{2}{\sqrt{k}} \text{ and }\Pr_{\x \sim \cD_1} \{\text{every
        $x_l$ is $0$}\} \leq \frac{1}{\sqrt{k}},$$ while matching
        moments up to degree $4$, i.e., $\forall i,j,m,n \in [k]$
\begin{align*}
\Ex_{\cD_0}[x_i] = \Ex_{\cD_1}[x_i] & &
\Ex_{\cD_0}[x_ix_jx_mx_n] = \Ex_{\cD_1}[x_ix_jx_mx_n] \\
\Ex_{\cD_0}[x_ix_j] = \Ex_{\cD_1}[x_ix_j] & & \Ex_{\cD_0}[x_ix_jx_m] = \Ex_{\cD_1}[x_ix_jx_m]
\end{align*}
\end{lemma}
\begin{proof}
For $\eps = \frac{1}{\sqrt{k}}$, take $\cD_1$ to be the following distribution:
\begin{enumerate}\item with probability $(1-\eps)$, randomly set exactly one of the bit to be 1 and all the other to be 0; \item with probability $\frac{\eps}{4}$, independently set every bit to be 1 with probability $\frac{1}{k^{1/3}}$; \item with probability $\frac{\eps}{4}$, independently set every bit to be 1  with probability $\frac{2}{k^{1/3}}$;
\item with probability $\frac{\eps}{4}$, independently set every bit to be 1  with probability $\frac{3}{k^{1/3}}$; \item
with probability $\frac{\eps}{4}$, independently set every bit to be 1  with probability $\frac{4}{k^{1/3}}$.
\end{enumerate}

The distribution $\cD_0$ is defined to be the following distribution with parameter $\eps_1,\eps_2,\eps_3,\eps_4$ to be specified later:
\begin{enumerate}
        \item with probability $1-(\eps_1+\eps_2+\eps_3+\eps_4)$, set every bit to be zero;
                       \item with probability $\eps_1$, independently set every bit to be 1 with probability $\frac{1}{k^{1/3}}$; \item with probability $\eps_2$, independently set every bit to be 1  with probability $\frac{2}{k^{1/3}}$;
                        \item with probability $\eps_3$, independently set every bit to be 1  with probability $\frac{3}{k^{1/3}}$;
                       \item with probability $\eps_4$, independently set every bit to be 1  with probability $\frac{4}{k^{1/3}}$.
\end{enumerate}

From the definition of $\cD_0,\cD_1$, we know that  $\Pr_{\x\sim \cD_0}[\text{every $x_i$ is $0$}] \geq 1-(\eps_1+\eps_2
+ \eps_3 + \eps_4)$ and $\Pr_{\x \sim \cD_1} [\text{every
        $x_i$ is $0$}] \leq \eps= \frac{1}{\sqrt{k}}$.

 It remains to determine each $\eps_i$. Notice that  the moment matching conditions can be expressed as a linear system over the
parameters $\eps_1,\eps_2,\eps_3,\eps_4$ as follows:
\begin{eqnarray*}
\sum_{i=1}^4 \eps_i (\frac{i}{k^{1/3}})  = (1-\eps)/k + \sum_{i=1}^4{\frac{\eps}{4}(\frac{i}{k^{\frac{1}{3}}}})\\
\sum_{i=1}^4 \eps_i (\frac{i}{k^{1/3}})^2  =  \sum_{i=1}^4{\frac{\eps}{4}(\frac{i}{k^{\frac{1}{3}}}})^2\\
 \sum_{i=1}^4 \eps_i (\frac{i}{k^{\frac{1}{3}}})^3  =  \sum_{i=1}^4{\frac{\eps}{4}(\frac{i}{k^{\frac{1}{3}}}})^3\\
 \sum_{i=1}^4 \eps_i (\frac{i}{k^{\frac{1}{3}}})^4  = \sum_{i=1}^4{\frac{\eps}{4}(\frac{i}{k^{\frac{1}{3}}}})^4.
 \end{eqnarray*}

We then   show that such a linear system
has a feasible solution $\eps_1,\eps_2,\eps_3,\eps_4>0$ and $\sum_{i=1}^{4} \eps_i\leq 2/ \sqrt{k}$ .

To prove this, by applying Cramer's rule,
\[
\eps_1 =\frac{\left| \begin{array}{cccc}
(1-\eps)/k + \sum_{i=1}^4{\frac{\eps}{4}(\frac{i}{k^{\frac{1}{3}}}}) & \frac{2}{k^{\frac{1}{3}}} & \frac{3}{k^{\frac{1}{3}}} & \frac{4}{k^{\frac{1}{3}}} \\
\sum_{i=1}^4{\frac{\eps}{4}(\frac{i}{k^{\frac{1}{3}}}})^2& \frac{4}{k^{\frac{2}{3}}} & \frac{9}{k^{\frac{2}{3}}} & \frac{16}{k^{\frac{2}{3}}} \\
\sum_{i=1}^4{\frac{\eps}{4}(\frac{i}{k^{\frac{1}{3}}}})^3& \frac{8}{k^{\frac{3}{3}}} & \frac{27}{k^{\frac{3}{3}}} & \frac{64}{k^{\frac{3}{3}}}
\\\sum_{i=1}^4{\frac{\eps}{4}(\frac{i}{k^{\frac{1}{3}}}})^4 & \frac{16}{k^{\frac{4}{3}}} & \frac{81}{k^{\frac{4}{3}}} & \frac{256}{k^{\frac{4}{3}}}\\
\end{array} \right|}
{\left|\begin{array}{cccc}
\frac{1}{k^{\frac{1}{3}}} & \frac{2}{k^{\frac{1}{3}}} & \frac{3}{k^{\frac{1}{3}}} & \frac{4}{k^{\frac{1}{3}}} \\
\frac{1}{k^{\frac{2}{3}}} & \frac{4}{k^{\frac{2}{3}}} & \frac{9}{k^{\frac{2}{3}}} & \frac{16}{k^{\frac{2}{3}}} \\
\frac{1}{k^{\frac{3}{3}}} & \frac{8}{k^{\frac{3}{3}}} & \frac{27}{k^{\frac{3}{3}}} & \frac{64}{k^{\frac{3}{3}}} \\
\frac{1}{k^{\frac{4}{3}}} & \frac{16}{k^{\frac{4}{3}}} & \frac{81}{k^{\frac{4}{3}}} & \frac{256}{k^{\frac{4}{3}}}\\
\end{array} \right|} \]
With some calculation using basic linear algebra, we get
\[
\eps_1 =
 \eps/4 + \frac{\left|
\begin{array}{cccc}
(1-\eps)/k  & \frac{2}{k^{\frac{1}{3}}} & \frac{3}{k^{\frac{1}{3}}} & \frac{4}{k^{\frac{1}{3}}} \\
0& \frac{4}{k^{\frac{2}{3}}} & \frac{9}{k^{\frac{2}{3}}} & \frac{16}{k^{\frac{2}{3}}} \\
0& \frac{8}{k^{\frac{3}{3}}} & \frac{3}{k^{\frac{3}{3}}} & \frac{64}{k^{\frac{3}{3}}}
\\
0 & \frac{16}{k^{\frac{4}{3}}} & \frac{3}{k^{\frac{4}{3}}} & \frac{256}{k^{\frac{4}{3}}}\\
\end{array} \right|}
{\left|\begin{array}{cccc}
\frac{1}{k^{\frac{1}{3}}} & \frac{2}{k^{\frac{1}{3}}} & \frac{3}{k^{\frac{1}{3}}} & \frac{4}{k^{\frac{1}{3}}} \\
\frac{1}{k^{\frac{2}{3}}} & \frac{4}{k^{\frac{2}{3}}} & \frac{9}{k^{\frac{2}{3}}} & \frac{16}{k^{\frac{2}{3}}} \\
\frac{1}{k^{\frac{3}{3}}} & \frac{8}{k^{\frac{3}{3}}} & \frac{27}{k^{\frac{3}{3}}} & \frac{64}{k^{\frac{3}{3}}} \\
\frac{1}{k^{\frac{4}{3}}} & \frac{16}{k^{\frac{4}{3}}} & \frac{81}{k^{\frac{4}{3}}} & \frac{256}{k^{\frac{4}{3}}}\\
\end{array} \right|} = \frac{1}{4\sqrt{k}} + O(\frac{1}{k^{\frac{2}{3}}}).
\]
For large enough $k$, we have $0\leq \eps_1\leq \frac{1}{2\sqrt{k}}$.
By similar calculation, we can bound $\eps_2,\eps_3,\eps_4$ by $\frac{1}{2\sqrt{k}}$. Overall, we have $\eps_1+\eps_2+\eps_3 + \eps_4 \leq 2/\sqrt{k}$

\end{proof}

We define a ``noisy'' version of $\cD_b$ ($b\in \{0,1\}$) below.
\begin{definition} For $b \in \bits$, define the distribution        $\tilde{\cD}_b$ on $\bits^k$ as follows:
        \begin{itemize}\itemsep=0ex
                \item First generate $x\in \bits^k$ according to ${\cD_b}$.
                 \item For each $i \in [k]$,
         \end{itemize}
                   $$ y_{i} = \begin{cases}
                            x_i  &  \text{ with
                            probability }  1- \frac{1}{k^2} \\
                            \text{uniform random bit } u_i & \text{ with
                            probability } \frac{1}{k^2}
                    \end{cases}$$
 \end{definition}
\begin{observation}
\label{obs:addnoise_tomoments}
 $\tilde{\cD_0}$ and $\tilde{\cD}_1$ also have matching moments up to degree $4$.
\end{observation}
\begin{proof}
Since the noise is defined to be an independent uniform random bit, when calculating moments of $y$, such as $\Ex_{\tilde{D}_b}[y_{i_1}y_{i_2}\cdots y_{i_d}]$, we can substitute  $y_i$ by $(1-\gamma)x_i +  \frac{1}{2}\gamma $. Therefore, a degree $d$ moment of $y$ can be expressed as a weighted sum
of moments of $x$ of degree up to $d$. Since $\cD_0$ and $\cD_1$ have matching
moments up to degree $4$, it follows that $\tilde{\cD}_0$ and $\tilde{\cD}_1$ also have the same property.
\end{proof}

The following simple lemma asserts that conditioning the two distributions
$\tilde{D}_0$ and $\tilde{D}_1$ on the same coordinate $x_j$ being
fixed to value $b$ results in conditional distributions that still have matching moments
up to degree $3$.
\begin{lemma}\label{lem:match} Given two distributions $\cP_0,\cP_1$
        on $\bits^k$ with matching moments up to degree $d$, for any multi-set $S$ of elements from $[k]$,
$|S|\leq d-1$, $j \in [k] $ and $c\in \bits$.
\[
        \Ex_{\calP_0}[\prod_{i\in S}x_i\mid x_j = c] = \Ex_{\calP_1}[\prod_{i\in S}x_i\mid x_j = c].
\]
\end{lemma}
\begin{proof}
For the case $c=1$ and any $b\in \bits$,
\[
         \Ex_{\calP_b}[ x_j \prod_{i\in S}x_i] = \Ex_{\calP_b}[\prod_{i\in
         S}x_i\mid x_j = 1] \Pr_{\cP_0}[x_j=1] = \Ex_{\calP_b}[\prod_{i\in
         S}x_i\mid x_j = 1]\Ex_{\cP_0}[x_j].
\]
Therefore,
\[
         \Ex_{\calP_0}[\prod_{i\in S}x_i\mid x_j = 1] = \frac{ \Ex_{\calP_0}[ x_j \prod_{i\in S}x_i]}{        \Ex_{\calP_0}[x_j]} =  \frac{ \Ex_{\calP_1}[ x_j \prod_{i\in S}x_i]}{        \Ex_{\calP_1}[x_j]} = \Ex_{\calP_1}[\prod_{i\in S}x_i\mid x_j = 1].
\]

For the case $c=0$, replace $x_j$ with $x'_j = 1- x_j $. It is easy to see that $\cP_0$ and $\cP_1$ still have matching moments and conditioning on $x_j = 0$ is the same as conditioning on $x'_j = 1$. Hence we can reduce to the case $c=1$.
\end{proof}

\subsection{The Dictatorship Test}
Let $R$ be a positive integer. Based on the distributions $\cD_0$ and $\cD_1$, we define the dictatorship test as follows:
\begin{center}
  \fbox{
    \parbox{6.0in}{
    \begin{enumerate} \itemsep=0ex
        \item Generate a random bit $b\in \{0,1\}$.
                \item Generate  $\vec{x} \in \bits^{kR}$ (which is also written as  $\{x_i^{(j)}\}_{i\in [k],j\in [R]}$) from $\cD_b^R$.
                \item For each $i \in [k], j \in [R]$,
                        $$ y_{i}^{(j)} = \begin{cases}
                                x_i^{(j)}  &  \text{ with
                                probability }  1- \frac{1}{k^2}; \\
                                \text{random bit} & \text{ with
                                probability } \frac{1}{k^2}.
                        \end{cases}$$

        \item Output the labelled example $(\vec{y},b)$.
                Equivalently, if $h$ denotes the halfspace, ACCEPT if $h(\vec{y})
                = b.$

\end{enumerate}
}
}
\end{center}

We can also view $y$  as being generated as follows:
i) With probability $\frac12$, generate  a negative sample from
distribution $\tilde{\cD}_0^R$;  ii)  With probability
$\frac12$, generate a positive sample from distribution $\tilde{\cD}_1^R$.

The dictatorship test has the following completeness and soundness properties.
\begin{theorem} {(completeness)}\label{thm:dictcp}
        For any $j \in [R]$, $h(\vec{y}) = \vee_{i=1}^k y_{i}^{(j)}$ passes
        with probability $ \geq 1-\frac{3}{\sqrt{k}}$.
\end{theorem}
\begin{theorem}
{(soundness)} \label{thm:testsoundness}
         Fix $\tau = \frac{1}{k^7}$ and $t = \frac{1}{\tau^2} (3\ln(1/\tau) + \ln R)  + \lceil 4k^2 \ln
        k \rceil \lceil \frac{4}{\tau^2} \ln(1/\tau)\rceil$.
        Let $h(\vec{x}) = \sgn(\iprod{\vec{w},\vec{y}} -
        \theta)$ be a halfspace such that $\bb_{t}(\vec{w}_i) \cap
        \bb_{t}(\vec{w}_j) = \emptyset$ for all $i,j \in [k]$.  Then the
        halfspace $h(\vec{y})$ passes the dictatorship test with
        probability at most $\frac{1}{2} + O(\frac{1}{k})$.
\end{theorem}
\begin{proof} (Theorem~\ref{thm:dictcp})
                If  $x$ is generated from $\cD_0^R$, we know that with
                probability at least $1-\frac{2}{\sqrt{k}}$,
                all the bits in $\{x_{1}^{(j)},
                x_{2}^{(j)},\ldots,x_{k}^{(j)}\}$
are set to $0$. By union bound, with probability at least $1-\frac{2}{\sqrt{k}}-\frac{1}{k}$,
$\{y_{1}^{(j)}, y_{2}^{(j)},\ldots,y_{k}^{(j)}\}$ are all set to $0$, in which
case the test  passes as $\vee_{i=1}^k y_{i}^{(j)}= 0$. If $x$ is
generated from $\cD_1^R$, we know that with probability at least
$1-\frac{1}{\sqrt{k}}$, one of the bits in $\{x_{1}^{(j)},
x_{2}^{(j)}, \ldots, x_{k}^{(j)}\}$ is set to $1$ and
by union bound one of $\{y_{1}^{(j)},
y^{(j)}_{2},\ldots,y_{k}^{(j)}\}$ is set to 1 with probability at least
$1-\frac{1}{\sqrt{k}}-\frac{1}{k}$, in which case the test  passes
since $\vee_{i=1}^k y_{i}^{(j)}= 1$. Overall, the test passes with probability at least $1-\frac{3}{\sqrt{k}}.$
\end{proof}
\subsection{Proof of Soundness (Theorem~\ref{thm:testsoundness}) }
We will prove the contrapositive statement of Theorem \ref{thm:testsoundness}: if some $h(\vec{y})$ passes the above dictatorship test with high probability,
then we can decode for each $\vec{w}_i$ ($i \in [k]$), a small  list of coordinates and at least
two of the lists will intersect.

The proof  is based on two key lemmas (Lemmas \ref{lem:commoninf}, \ref{lem:numinf}).
The first lemma states that if a halfspace passes the test with good probability, then two of its critical index sets $C_{\tau}(\vec{w}_i),C_{\tau}(\vec{w}_j)$ must intersect. This would immediately imply Theorem~\ref{thm:testsoundness} if $c_{\tau}$ is less than $t$.  The second lemma states that every halfspace can be approximated by another halfspace with critical index less than $t$; so we can assume that $c_\tau$ is small without loss of generality.

Let $h(\vec{y})$ be a halfspace function on $\bits^{kR}$ given by
$h(\vec{y}) = \sgn(\iprod{\vec{w},\vec{y}}-\theta)$.  Equivalently, $h(\vec{y})$ can be written as
\begin{align*}
h(\vec{y}) = \sgn\Big(\sum_{j \in [R]} \iprod{
\vec{w}^{(j)},\vec{y}^{(j)} } - \theta\Big)  = \sgn\Big(\sum_{i \in [k]} \iprod{
\vec{w}_i,\vec{y}_i } - \theta\Big)\ ,
\end{align*}
where $\vec{w}^{(j)} \in \R^k$ and $\vec{w}_i \in \R^R$.
\ignore{\begin{lemma} $\tilde{D_0}$ and $\tilde{D_1}$ has matching moments up to
degree 4.
\end{lemma}
\begin{proof}
        Introduce a random vector $z_i^{(j)}$ as a indicator of whether $y_i^{(j)}$ is generated as a random bit or a copy of $x_i^{(j)}$. If suffice to show for every set up of $z_i^{j}$, the moments match up to degree 4.
\end{proof}
}

 \begin{lemma} {(Common Influential Coordinates)} \label{lem:commoninf}
For $\tau =\frac{1}{k^7}$,  let $h(\vec y)$ be a halfspace such that for all $i \neq j \in [k]$,  we have $C_{\tau}(\vec{w}_i) \cap
        C_{\tau}(\vec{w}_j) = \emptyset$ .  Then
        \begin{equation*}
                \Big| \Ex_{\tilde{\cD}_0^R}[h(\vec{y})] -
                \Ex_{\tilde{\cD}_1^R}[h(\vec{y})] \Big| \leq
               O\Big(\frac{1}{k}\Big)\ .
\end{equation*}
\end{lemma}
\begin{proof}Fix the following notation,
\begin{align*}
        \vec{l}_i = \truncate(\vec{w}_i, C_{\tau}(\vec{w}_i)) & &
        \vec{s}_i =
        \vec{w}_i - \w^C_i \\ \vec{y}_i^{C} = \truncate(\vec{y}_i,
        C_{\tau}(\vec{w}_i)) & &
        \vec{y}^{C} = \vec{y}_1^{\hh},\vec{y}_2^{\hh},\ldots,\vec{y}_k^{\hh}
        \\ \vec{s} = \vec{s}_1,\vec{s}_2,\ldots,\vec{s}_k & &  \vec{l} = \vec{l}_1,\vec{l}_2,\ldots,\vec{l}_k.
\end{align*}
We can rewrite the halfspace $h(\vec y)$ as
$ h(\vec{y}) = \sgn\Big(\iprod{
\vec{l},\vec{y}^{\hh} } +  \iprod{
\vec{s},\vec{y} }  - \theta\Big)$.
Let us first normalize the halfspace $h(\vec{y})$ so that
$\sum_{i \in [k]} \|\vec{l}_i\|^2 = 1$.
We now condition on a possible fixing of the vector $\vec{y}^{\hh}$.
Under this conditioning and for $\vec{y}$ chosen randomly from the distribution $\tilde{\cD}^R_{0}$, define the family of
ensembles $\cA =
\bvec{A}{1},\ldots,\bvec{A}{R}$ as follows:
\begin{eqnarray*}
        \bvec{A}{j} = \{ y^{(j)}_i |\  i \in [k] \text{ for which } j \notin C_{\tau}(\vec{w}_i) \}
\end{eqnarray*}
Similarly define the ensemble $\cB =
\bvec{B}{1},\ldots,\bvec{B}{R}$ using $\vec{y}$ chosen randomly from the distribution
$\tilde{\cD}^R_{1}$.  Further let us denote $\bvec{l}{j} =
(l_{1}^{(j)},\ldots,l_{k}^{(j)})$.  Now we apply the invariance
principle (Theorem \ref{thm:invariance}) to the ensembles $\cA,\cB$ and the linear function $\vec{l}$.  For each $j \in [R]$, there is at most one
coordinate $i \in [k]$ such that $j \in C_{\tau}(\vec{w}_i)$.  Thus,
conditioning on $\vec{y}^{\hh}$ amounts to fixing of at most
one variable $y_i^{(j)}$ in each column $\{y_i^{(j)}\}_{i \in [k]}$.  By Lemma
\ref{lem:match},
since $\tilde{\cD_0}$ and $\tilde{\cD_1}$ have matching moments up to degree $4$, we get that $\bvec{A}{j}$ and $\bvec{B}{j}$ have matching moments up to degree $3$. Also notice that $\max_{j
        \in [R],i\in[k]} |l_{i}^{(j)}| \leq \tau \| \vec{l}_i\|_2\leq \tau
\|\vec{l}\|_2$ (as $\vec{l}_i$ is a $\tau$-regular) and each  $y_{i}^{(j)}$ is set to be a random unbiased bit with probability $\frac{1}{k^2}$; by Lemma \ref{lem:spread}, the linear function
$\vec{l}$ and the ensembles $\cA$, $\cB$ satisfy the following spread
property for every $\theta' \in \R$:
        \begin{align*}
                \Pr_{\cA} \Big[ \vec{l}(\cA) \in [\theta' - \alpha,
                \theta'+\alpha] \Big] \leq c(\alpha)\\  \Pr_{\cB} \Big[ \vec{l}(\cB) \in [\theta' - \alpha,
        \theta'+\alpha] \Big] \leq c(\alpha),
        \end{align*}
        where $c(\alpha) \leq 8\alpha k + 4\tau k +
        2e^{-\frac{1}{2\tau^2 k^4}}$ (by setting
        $\gamma = \frac{1}{k^2}$ and $|b-a| =2\alpha$ in Lemma
        \ref{lem:spread}).
 Using the invariance principle (Theorem
\ref{thm:invariance}) this implies:
\begin{multline} \label{eqn:tobound}
        \Big| \Ex_{\cA}\Big[\sgn\Big( \iprod{
\vec{s},\vec{y}^{\hh} } + \sum_{j \in [R]} \iprod{
\bvec{l}{j},\bvec{A}{j} }  - \theta\Big)\Big| \vec{y}^{\hh} \Big] -\\
\Ex_{\cB}\Big[\sgn\Big( \iprod{
\vec{s},\vec{y}^{\hh} } + \sum_{j \in [R]} \iprod{
\bvec{l}{j},\bvec{B}{j} }  - \theta\Big)\Big| \vec{y}^{\hh} \Big]
\Big| \\\leq O\Big(\frac{1}{\alpha^4}\Big) \sum_{i \in [R]}
                \|\bvec{l}{i}\|_1^4 + 2 c(\alpha)
\end{multline}
 By definition of the critical index, we have  $\max_{j
\in [R]} l_{i}^{j} \leq \tau \| \vec{l}_{i}\|_2$.  Using this, we can bound
$\sum_{i \in [R]}
                \|\bvec{l}{i}\|_1^4$ as follows:
\begin{multline*}
                 \sum_{j \in [R]} \| \bvec{l}{j} \|_1^4
                \leq k^4 \sum_{i \in [k]}\sum_{j \in [R]}   \| l_{i}^{(j)} \|^4
                \leq k^4 \sum_{i \in [k]} \Big(\max_{j \in [R]}
                |l_{i}^{(j)}|^2 \Big)  \| \vec{l}_i\|_2^2 \\  \leq
                k^4 \tau^2 \sum_{i \in [k]} \| \vec{l}_i\|_2^2
                \leq  k^4 \tau^2 \| \vec{l}\|_2^2 \leq \frac{1}{k^{10}}.
\end{multline*}
In the final inequality in above calculation, we used the fact that $\tau = \frac{1}{k^7}$ and $\| \vec{l}\|_2 = 1$.
Let us choose $\alpha = \frac{1}{k^2}$ and \eqref{eqn:tobound} is therefore bounded by $O(1/k)$  for all settings of
$\vec{y}^{\hh}$.  Averaging over all settings of $\vec{y}^{\hh}$ we get that
\[
        \Big| \Ex_{\tilde{\cD}_0^R}[h(\vec{y})] -
                \Ex_{\tilde{\cD}_1^R}[h(\vec{y})] \Big|\leq O\left(\frac{1}{k}\right). \]
\end{proof}

The above lemma asserts that unless some two vectors $\vec{w}_i,
\vec{w}_j$ have a {\it common influential coordinate}, the halfspace
$h(\vec{y})$ cannot distinguish between $\tilde{\cD}_0^R$ and
$\tilde{\cD}_1^R$. Unlike with the traditional notion of influence, it is unclear whether
the number of coordinates in $C_{\tau}(\vec{w}_i)$ is small.  The
following lemma yields a way to get around this.

\begin{lemma} {(Bounding the number of influential coordinates)}
        \label{lem:numinf}
        Let $t$ be set as in Theorem~\ref{thm:testsoundness}. Given a halfspace $h(\vec{y})$ and $\rr \in [k]$ such that
        $|C_{\tau}(\vec{w}_\rr)| > t$, define $\tilde{h}(\vec{y}) =
        \sgn(\sum_{i \in [k]} \iprod{\tilde{\vec{w}}_i,\vec{y}_i} -
        \theta)$ as follows: $\tilde{\vec{w}}_\rr =
        \truncate(\vec{w}_\rr, \bb_t(\vec{w}_\rr))$ and $\tilde{\vec{w}}_i =
        \vec{w}_i$ for all $i \neq \rr$.
        Then,
        \begin{align*}
                \Big| \Ex_{ \tilde{\cD}_0^R}[\tilde{h}(\vec{y})] -
                \Ex_{ \tilde{\cD}_0^R}[{h}(\vec{y})] \Big| \leq
                {\frac{1}{k^2}} \mbox { and } \Big|
                \Ex_{ \tilde{\cD}_1^R}[\tilde{h}(\vec{y})] -
                \Ex_{ \tilde{\cD}_1^R}[{h}(\vec{y})] \Big| \leq
                {\frac{1}{k^2}}.
        \end{align*}
\end{lemma}
\begin{proof} Without loss of generality, we assume $\rr=1$ and  $|w_{1}^{(1)}| \geq
        |w_{1}^{(2)}|\geq  \cdots \geq |w_{1}^{(R)}|$.  In particular,
        this implies $\bb_{t}(\vec{w}_1) = \{1,\ldots,t\}$.  Set $T =\lceil 
        4k^2 \ln k \rceil$.  Define the
        subset $G$ of $\bb_{t}(\vec{w}_1)$ as $$G = \{g_i \ |\  g_i = 1 + i \lceil(4/\tau^2)
        \ln(1/\tau)\rceil,  0\leq i \leq T \}.$$
        Therefore, by Lemma~\ref{lem:st}, $|w_1^{(g_i)}|$ is a geometrically decreasing
        sequence such that $|w_1^{(g_{i+1})}|\leq |w_1^{(g_i)}|/3$.  Let $H = \bb_t(\vec{w}_1)\setminus G$.  Fix the following notation:
        \begin{equation*}
        \vec{w}_1^{G} = \truncate(\vec{w}_1,G), \hspace{2ex}
        \vec{w}_1^{H} =\truncate(\vec{w}_1,H),  \hspace{2ex}
        \vec{w}_1^{>t}= \truncate(\vec{w}_1,\{t+1,\ldots,n\}).
        \end{equation*}
        Similarly, define the vectors $\vec{y}_1^G,\vec{y}_1^H,
        \vec{y}_1^{>t}$. We now rewrite the halfspace functions $h(\vec{y})$ and
        $\tilde{h}(\vec{y})$ as:
        $$h(\vec{y}) = \sgn \Big(\sum_{i =2}^k
                \iprod{\vec{w}_i,\vec{y}_i} + \iprod{\vec{w}_1^{G} , \vec{y}_1^G} +
                \iprod{\vec{w}_1^{H} , \vec{y}_1^H}
           + \iprod{\vec{w}^{>t}_1 ,  \vec{y}_1^{>t}} - \theta  \Big) $$

        $$\tilde{h}(\vec{y}) = \sgn \Big(\sum_{i =2}^k
        \iprod{\vec{w}_i,\vec{y}_i} + \iprod{\vec{w}_1^{G} , \vec{y}_1^G} +
        \iprod{\vec{w}_1^{H} , \vec{y}_1^H}
        - \theta  \Big)\ .$$
               Notice that for any $\vec{y}$, $h(\vec{y})\ne \tilde{h}(\vec{y})$ implies
                \begin{equation}\label{eqn:bp}
                                        \big|\sum_{i =2}^k
                                \iprod{\vec{w}_i,\vec{y}_i} + \iprod{\vec{w}_1^{G} , \vec{y}_1^G} +
                                \iprod{\vec{w}_1^{H} , \vec{y}_1^H} - \theta \big| \leq |\iprod{\vec{w}_1^{>t} , \vec{y}_1^{>t}} |.
                       \end{equation}

By Lemma \ref{lem:st}, we know that
        \begin{multline*}
                |w_1^{(g_T)}|^2 \geq
                \frac{\tau^2}{(1-\tau^2)^{t-g_T}}
                \|\vec{w}_{1}^{>t}\|^2_2
                 \geq \frac{\tau^2}{(1-\tau^2)^{\frac{1}{\tau^2} (3\ln(1/\tau) + \ln R) }} \|\vec{w}_{1}^{>t}\|^2_2 \geq \frac{R}{\tau} \|\vec{w}_{1}^{>t}\|^2_2.
\end{multline*}
Using the fact that $R \|\vec{w}_{1}^{>t}\|_2^2\geq \|\vec{w}_{1}^{>t}\|^2_1  $, we can get that
$\|\vec{w}_1^{>t}\|_1\leq \sqrt{\tau} |w_{1}^{(g_T)}| \leq \frac{1}{6}|w_{1}^{(g_{T})}|$.
Combining the above inequality with (\ref{eqn:bp}) we see that,
        \begin{align*}
        \Prx_{\tilde{\cD}_0^R} \Big[h(\vec{y})  \neq                \tilde{h}(\vec{y})\Big]
                  &\leq            \Prx_{\tilde{\cD}_0^R}\Big[| \sum_{i =2}^k
                                        \iprod{\vec{w}_i,\vec{y}_i} +
                                        \iprod{\vec{w}_1^{G} ,
                                        \vec{y}_1^G} +
                                        \iprod{\vec{w}_1^{H} ,
                                        \vec{y}_1^H} -\theta  |
                \leq |\iprod{\vec{w}_1^{>t} , \vec{y}_1^{>t}} |\Big]\\
         & \leq   \Prx_{\tilde{\cD}_0^R}\Big[\big| \sum_{i =2}^k
                        \iprod{\vec{w}_i,\vec{y}_i} +
                        \iprod{\vec{w}_1^{G} , \vec{y}_1^G} +
                        \iprod{\vec{w}_1^{H} , \vec{y}_1^H}
                        -\theta\big|\leq \frac{|w_{1}^{(g_T)}|}{6}\Big]\\
        &= \Prx_{\tilde{\cD}_0^R} \Big[ \iprod{\vec{w}_1^{G} , \vec{y}_1^G}  \in
        [\theta'-\frac{1}{6}|w_{1}^{(g_T)}|,
        \theta'+\frac{1}{6}|w_{1}^{(g_T)}|]  \Big]
        \end{align*}
        where $\theta' = -\sum_{i =2}^k
        \iprod{\vec{w}_i,\vec{y}_i} -
        \iprod{\vec{w}_1^{H} , \vec{y}_1^H}
        + \theta$.  For any fixing of the value of $\theta' \in \R$,
       it induces a certain distribution on $\vec{y}_1^G$.  However,
        the $\frac{1}{k^2}$ noise introduced in $\vec{y}_1^G$ is
        completely independent.  This corresponds to the setting of
        Lemma \ref{lem:geometricsmallball1}, and hence we can bound the above probability
        by $\left(1-\frac{1}{2k^2}\right)^T \leq \frac{1}{k^2}$.
        The result  follows from averaging over all values of $\theta'$.
\end{proof}

With the two lemmas above,  we now prove the soundness property.

\begin{proof} (Theorem~\ref{thm:testsoundness})
The probability of success of $h(\vec{y})$ is
        given by  $ \frac{1}{2} +\frac{1}{2} \big (\Ex_{
\tilde{\cD}_1^R}[h(\vec{y}) ] - \Ex_{
\tilde{\cD}_0^R}[h(\vec{y})]\big)$.
\noindent Therefore, it suffices to show that $
         \Big| \Ex_{\tilde{\cD}_0^R}[h(\vec{y})] -
      \Ex_{\tilde{\cD}_1^R}[h(\vec{y})] \Big| =
     O(\frac{1}{k}).$

     \smallskip \noindent Define $I = \{r\mid C_{\tau}(\w_r) \geq t
     \}$. We discuss the following two cases.

\smallskip
\noindent 1.  $I = \emptyset$; i.e., $\forall i \in [k]$, $C_{\tau}(\w_i) \leq t$. Then for all $i,j$,                  $\bb_{t}(\vec{w}_i) \cap
    \bb_{t}(\vec{w}_j) = \emptyset$ implies $C_{\tau}(\vec{w}_i) \cap
    C_{\tau}(\vec{w}_j) = \emptyset$. By Lemma \ref{lem:commoninf}, we thus
have
$\Big| \Ex_{\tilde{\cD}_0^R}[h(\vec{y})] -
        \Ex_{\tilde{\cD}_1^R}[h(\vec{y})] \Big| =
       O(\frac{1}{k})$.

\smallskip
  \noindent   2.  $I\ne \emptyset$. Then for all $r\in I$, we set
     $\tilde{\vec{w}}_\rr = \truncate(\vec{w}_\rr,
     \bb_t(\vec{w}_\rr))$ and replace $\vec{w}_\rr$ with
     $\tilde{\vec{w}}_\rr$ in $h$ to get a new halfspace $h'$. Since
     such replacements occur at most $k$ times and by Lemma
     \ref{lem:numinf} every replacement changes the output of the
     halfspace on at most $\frac{1}{k^2}$ fraction of examples, we can
     bound the overall change by $k\times \frac{1}{k^2} =
     \frac{1}{k}$. That is
   \begin{eqnarray}\label{eqn:diff1:dict}
            \Big| \Ex_{ \tilde{\cD}_0^R}[h'(\vec{y})] -
            \Ex_{ \tilde{\cD}_0^R}[{h}(\vec{y})] \Big| \leq
            {\frac{1}{k}}, \qquad         \Big|
            \Ex_{ \tilde{\cD}_1^R}[h'(\vec{y})] -
            \Ex_{ \tilde{\cD}_1^R}[{h}(\vec{y})] \Big| \leq
            {\frac{1}{k}}.
    \end{eqnarray}
Also notice that for $h'$ and all $\rr \in [k]$,  the critical index
of $\tilde{\vec{w}}_\rr$ (i.e., $|C_{\tau}(\tilde{\vec{w}}_\rr)|$) is less than $t$. This reduces the
problem to Case 1, and we conclude
$         \Big| \Ex_{\tilde{\cD}_0^R}[h'(\vec{y})] -
        \Ex_{\tilde{\cD}_1^R}[h'(\vec{y})] \Big| =
       O(1/k)$.
Along with (\ref{eqn:diff1:dict}) this finishes the proof of Theorem~\ref{thm:testsoundness}.

\end{proof}

\subsection{Reduction from \kuniquelabelcover}\label{sec:hardness}
With the dictatorship test defined, we now describe briefly a reduction from \kuniquelabelcover
problem to agnostic learning of monomials, thus showing Theorem
\ref{thm:main} under the Unique Games Conjecture (Conjecture
\ref{ugconj}). Although our final hardness result only assumes $\P \ne \NP$, we describe the reduction to \kuniquelabelcover for the purpose of illustrating the main idea of our proof.

Let $\cL(G(V,E),R,R,\{\pi^{v,e}|v \in V, e \in E\})$ be an instance of
\kuniquelabelcover.  The reduction is defined in  Figure~\ref{fig:ugreduction}. It will produce a distribution over
labeled examples:  $(\y,b) $ where $\y \in \{0,1\}^{|V|\times R}$ and label  $b\in
\{0, 1\}$.  We will index the coordinates
of $\vec{y}\in \{0,1\}^{|V|\times R}$  by $y_{w}^{(i)}$ (for $w\in V, i\in R)$ and denote $\y_w$ (for $w\in V)$  to be
the vector $(y_w^{(1)},y_w^{(2)},\ldots,y_w^{(R)})$.
\begin{figure}\label{fig:ugreduction}
{\center
  \fbox{
    \parbox{6.0in}{
    \begin{enumerate}\itemsep=0ex
\item Sample an edge $e = (v_1,\ldots,v_k) \in E$.
\item Generate a random bit $b \in \bits$.
\item Sample $\vec{x} \in \bits^{kR}$ from $\tilde{\cD}_b^R$.
\item Define $\vec{y} \in \bits^{|V| \times R}$ as follows:
        \begin{enumerate}
                \item For each $v \notin \{v_1,\ldots,v_k\}$, $\vec{y}_v = \vec{0}$.
                \item For each $i \in [k]$ and $j \in [R]$, $y_{v_i}^{(j)} = x_i^{(\pi^{v_i,e}(j))} .$
        \end{enumerate}
\item Output the example $(\vec{y},b)$.
\end{enumerate}
}
}
}
\caption{Reduction from \kuniquelabelcover}
\end{figure}
\\
\paragraph{Proof of Theorem \ref{thm:main} assuming Unique Games
Conjecture}
Fix $k = \frac{10}{\eps^2}$, $\eta = \frac{\eps^3}{100}$ and a
positive integer $R >
\lceil (2k)^{\frac{1}{\eta^2}}\rceil$ for which Conjecture
\ref{ugconj} holds.

\noindent \textbf{Completeness:} Suppose that $\Lambda:V
\to [R]$ is a labeling that {\em strongly} satisfies $1-k \eta$ fraction
of the edges.  Consider disjunction $h(\vec{y}) = \bigvee_{v\in V}
y_{v}^{(\Lambda(v))}$.
For at least $1-k\eta$ fraction of edges $e  = (v_1,v_2,\ldots,v_k) \in E$, $\pi^{v_1,e}(\Lambda(v_1)) = \cdots =
\pi^{v_k,e}(\Lambda(v_k)) = \rr$.  Let us fix such a choice of edge $e$ in step 1.  As all coordinates of $\vec{y}$
outside of $\{ \vec{y}_{v_1},\ldots,\vec{y}_{v_k}\}$ are set
to $0$ in step $4(a)$, the disjunction reduces to $\vee_{i\in [k]}
y_{v_i}^{(\Lambda(v_{i}))} = \vee_{i \in [k]} x_{i}^{(\rr)}$. By Theorem \ref{thm:dictcp},
such a disjunction agrees with every $(\vec{y},b)$ with probability at least $1-\frac{3}{\sqrt{k}}$.
Therefore $h(\vec{y})$ agrees with a random example with probability at least $(1-\frac{3}{\sqrt{k}})(1-k\eta) \geq 1- \frac{3}{\sqrt{k}}-k\eta \geq 1-\eps$.

\noindent \textbf{Soundness:} Suppose there exists a halfspace  $h(\vec{y}) =
\sum_{v\in V} \iprod{\w_{v}, \y_{v}}$ that agrees with more than
$\frac{1}{2} + \eps \geq \frac{1}{2} + \frac{1}{\sqrt{k}}$ fraction of the examples.
Set $t=k^{14} (3\ln(k^7) + \ln R)  + \lceil 4k^{14} \ln k^7 \rceil 
         \cdot \lceil 4k^2 \ln k\rceil = O\big(k^{16}\ln{R}\big)$ (same as in
        Theorem~\ref{thm:testsoundness}).
Define the labeling $\Lambda$ using the following strategy : for each
vertex $v\in V$ randomly pick a label
from $\bb_t(\w_{v})$.

By an averaging argument, for at least $\frac{\eps}{2}$ fraction of the edges
$e\in E$ generated
in step 1 of the reduction, $h(\vec{y})$ agrees with the examples corresponding to $e$ with probability at least $\frac12 +
\frac{\eps}{2}$.
We will refer to such edges as {\it good}.
By Theorem~\ref{thm:testsoundness} for each  {\it good} edge $e \in
E$, there exists $i,j \in [k]$,  such that $\pi^{v_i,e}\big(\bb_{t}(\vec{w}_{v_i})\big)\cap
\pi^{v_j,e}\big(\bb_{t}(\vec{w}_{v_j})\big)\ne \emptyset$.  Therefore the edge
$e \in E$ is {\it weakly} satisfied by the labeling $\Lambda$ with
probability at least $\frac{1}{t^2}$.  Hence, in expectation the
labeling $\Lambda$ {\it weakly} satisfies at least $\frac{\eps}{2}
\cdot \frac{1}{t^2} = \Omega(\frac{1}{k^{33} \ln^2{R}}) \geq
\frac{2k^2}{R^{\eta/4}}$ fraction of the edges (by the
choice of $R$ and $t$).

\section{Reduction from Label Cover}\label{sec:truehd}
In this section, we describe a reduction from a  \klabelcover with an additional {\em smoothness} property to the problem of agnostic learning of disjunctions by halfspaces. This will give us Theorem~\ref{thm:main} without assuming the Unique Games Conjecture.

\subsection{Smooth \klabelcover}
Our reduction use the following hardness result for \klabelcover(Definition~\ref{def:klabelcover}) with the additional smoothness property.

\begin{theorem}\label{thm:sml}
        There exists a constant $\gamma> 0$ such that for any integer parameter $J,u\geq 1$, it is \emph{NP}-hard to distinguish between the following two types of \klabelcover $\cL(G(V,E),M,N,\{\pi^{v,e} | e\in E, v\in e\})$ instances with $M = 7^{(J+1)u}$ and $N= 2^{u}7^{Ju}$:
        \begin{enumerate}
                \item (Strongly satisfiable instances) There is some labeling that strongly satisfies every hyperedge.
                \item (Instances that are not $2k^2 2^{-\gamma u}$-weakly satisfiable) There is no labeling that weakly satisfies at least $2k^2 2^{-\gamma u}$ fraction of the hyperedges.
        \end{enumerate}
        In addition, the \klabelcover instances have the following properties:
        \begin{itemize}
              \item (Smoothness) for a fixed vertex $v$ and a randomly picked hyperedge containing $v$,
                \[
                        \forall i,j\in [M], \Pr[\pi^{v,e}(i) = \pi^{v,e}(j)] \leq 1/J.
                \]
                \item For any mapping $\pi^{v,e}$ and any number $i \in [N]$, we have $|(\pi^{v,e})^{-1}(i)|\leq d = 4^u$; i.e., there are at most $d =
4^u$ elements in $[M]$ that are mapped to the same number in $[N]$.
        \end{itemize}
\end{theorem}
The proof of the above theorem can be found in Appendix~\ref{sec:sm}.

In the rest of the paper, we will set $u = k$ and therefore $d = 4^k$. Also we set the smoothness parameter $J= d^{17} = 4^{17k}$.

\subsection{Reduction from Smooth \klabelcover}

The starting point is a smooth \klabelcover $\cL(G(V,E),M,N,\{\pi^{v,e} | e\in E, v\in e\})$ with $M = 7^{(J+1)u}$ and $N= 2^{u}7^{Ju}$  as described in Theorem \ref{thm:sml}.
Figure~\ref{fig:lcreduction} illustrates the reduction from \klabelcover $\cL(G(V,E),N,M,\{\pi^{v,e} | e\in E, v\in e\})$ that given an instance of \klabelcover $\cL$ produces a random labeled example. We refer to the obtained distribution on examples as $\calE$.
\begin{figure}
\label{fig:lcreduction}
{\center
  \fbox{
    \parbox{\textwidth}{
\begin{itemize}
\item Pick a hyperedge $e = (v_1,v_2,\ldots,v_k) \in E$ with
        corresponding projections $\pi^{v_1,e},\ldots,\pi^{v_k,e} : [M] \to [N]$.
\item Generate a random bit $b \in \bits$.
\item Sample $\vec{x} \in \bits^{kN}$ from $\cD_b^N$.
\item Generate $\vec{y} \in \bits^{|V| \times M}$ as follows:
        \begin{enumerate}
                \item For each $v \notin e$, $\vec{y}_v = \vec{0}$.
                \item For each $i \in [k]$,  set $\vec{y}_{v_i} \in
                        \bits^{M}$ as
                        follows:
                        $$ y_{v_i}^{(j)} = \begin{cases}
                                x_i^{(\pi^{v_i,e} (j))}  &  \text{ with
                                probability }  1- \frac{1}{k^2} \\
                                \text{random bit} & \text{ with
                                probability } \frac{1}{k^2}
                        \end{cases}$$
        \end{enumerate}
\item Output the example $(\vec{y},b)$ or equivalently
        ACCEPT if $h(\vec{y}) = b$.
\end{itemize}
    }
  }
}
\caption{Reduction from \klabelcover}
\end{figure}
\subsection{Proof of Theorem \ref{thm:main}}

We claim that our reduction has the following completeness and soundness properties.
\begin{theorem}\label{thm:main-hardness}
\begin{itemize}
\item {\sc Completeness:} If $\cL$ is a strongly-satisfiable
        instance of smooth \klabelcover, then there is a disjunction
        that agrees with a random example from $\calE$ with probability at least $1-O(\frac{1}{\sqrt{k}})$.
\item {\sc Soundness:} If $\cL$ is not $2k^2 2^{-\gamma k}$-weakly satisfiable and is {\it smooth} with parameters $J = 4^{17k}$ and $d = 4^k$,
        then there is no halfspace that agrees with a random example from $\calE$ with probability more than
        $\frac{1}{2} + O(\frac{1}{\sqrt{k}})$.
\end{itemize}
\end{theorem}
Combining the above theorem with Theorem~\ref{thm:sml} we get that for $k=O(1/\eps^2)$, we obtain our main result: Theorem~\ref{thm:main}.

It remains to check the correctness of the completeness and soundness claims in Theorem~\ref{thm:main-hardness}. First let us prove the completeness property.

\begin{proof}(Proof of Completeness)
Let $\Lambda$ be the labeling that strongly satisfies $\cL$. Consider disjunction $h(\vec{y}) = \bigvee_{v\in V}
y_{v}^{(\Lambda(v))}$.  Let $e= (v_1,v_2,\ldots,v_k)$ be any hyperedge and let $\calE_e$ be the distribution $\calE$ restricted to the examples generated for $e$. With probability at least $1-1/k$,
$y_{v_i}^{\Lambda(v_i)} = x_{i}^{\pi^{v_i,e}(\Lambda(v_i))}$  for
every $i \in [k]$. As $e$ is strongly satisfied by $\Lambda$, for all $i,j \in [k]$, $\pi^{v_i,e}(\Lambda(v_i)) = \pi^{v_j,e}(\Lambda(v_j))$. Therefore, as in the proof of Theorem \ref{thm:dictcp}, we obtain that $h(\vec{y})$ agrees with a random example from $\calE_e$ with probability at least $1-O(1/\sqrt{k})$. Labeling $\Lambda$ strongly satisfies all edges and therefore we obtain that $h(\vec{y})$ agrees with a random example from $\calE$ with probability at least $1-O(1/\sqrt{k})$. \end{proof}

The more complicated part is the soundness property which we prove in Section~\ref{sec:tsd}.

\subsection{Soundness Analysis}\label{sec:tsd}
\paragraph{Proof Idea}
The main idea is similar to the proof of Theorem~\ref{thm:testsoundness} although it is more technically involved. Notice that the reduction in Figure~\ref{fig:lcreduction} produces examples such  that $y_{v_i}^{j_1}$ , $y_{v_i}^{j_2}$ are ``almost identical'' copies when $\pi^{v_i,e}(j_1)=\pi^{v_i,e}(j_2)$.   Further for different edges $e$, the coordinates of $y$ will be grouped in different ways, such that each group will have almost identical copies.

To handle these  additional complications, the first step of the proof is to show that  almost all the hyperedges in smooth \klabelcover satisfy a certain ``niceness" property.  After that we generalize the   proofs of Lemma~\ref{lem:commoninf} and Lemma~\ref{lem:numinf} under the weaker assumption that most of the hyperedges are ``nice''.

The formal definition of ``niceness'' and the proof that most of the edges are ``nice" appear  in Section~\ref{sec:smooth}.
The  generalization of Lemma~\ref{lem:commoninf}  appears in Section~\ref{sec:cinf}. The generalization of Lemma~\ref{lem:numinf} appears in  Section~\ref{sec:ninf}. All these results are put together into a proof of  Theorem~\ref{thm:main-hardness} in Section~\ref{sec:all}.

\subsubsection{Most of the edges are ``nice''}\label{sec:smooth}
Let $h(\vec{y})$ be a halfspace that agrees with more than
$\frac{1}{2}+\frac{1}{\sqrt{k}}$-fraction of the examples.  Suppose,
\begin{eqnarray*}
        h(\vec{y}) = \sgn\Big( \sum_{v \in V} \iprod{\vec{w}_v,
        \vec{y}_v} - \theta\Big).
\end{eqnarray*}

Let $\tau = \frac{1}{k^{13}}$ and let
$$\vec{s}_v = \truncate(\vec{w}_v, C_{\tau}(\vec{w}_v)), \ \ \
        \vec{l}_v =
        \vec{w}_v - \vec{s}_v.
$$
\begin{definition}
        A vertex $v \in V$ is said to be {\it $\beta$-nice} with respect to a
        hyperedge $e \in E$ containing it if
        \begin{eqnarray*}
        \sum_{i \in [N]} \Big( \sum_{j \in \pi^{-1}(i)} |l_{v}^{(j)}|
        \Big)^4 \leq \beta \|\vec{l}_v\|^4_{2},
        \end{eqnarray*}
        where $\pi : [M] \to [N]$ is the projection associated with
        vertex $v$ and hyperedge $e$.  
        A hyperedge $e = (v_1,v_2,\ldots,v_k)$ is $\beta$-nice, if for
        every $i \in [k]$, the vertex $v_i$ is $\beta$-nice with
        respect to $e$.
 \end{definition}

\begin{lemma} \label{lem:vertexnice}
        The fraction of $2\tau$-nice hyperedges in $E$ is at least $1 - O(1/k)$.

\end{lemma}
\begin{proof} By definition, we know that  $\vec{l}_v$ is $\tau$-regular vector.
Denote $I_v =\{i\mid \frac{(l_v^{(i)})^2}{\|\vec{l}_v\|_2^2}\geq
\frac{1}{d^8}\}$. By definition $|I|\leq d^{8}$.  Notice there are at
most $d^{16}$ pairs of values in $I \times I$. By the smoothness property
of the \klabelcover instance, for any vertex $v$,
 at least $1- \frac{d^{16}}{J}$ fraction of
the hyperedges incident on $v$ have the following property: for any $i,j\in I_v$, $\pi^{v,e}(i)\ne \pi^{v,e}(j)$.
If all the vertices in a hyperedge have this property we call it a {\em good} hyperedge. By an averaging argument, we know that among all hyperedges at least  $1- \frac{kd^{16}}{J} = 1-\frac{k}{4^k}\geq 1-O(\frac{1}{k})$ fraction is {\em good}.

We will show all these {\em good} hyperedges are also $2\tau$-nice. For a given
{\em good} hyperedge $e$,  a vertex $v\in
e$, $\pi = \pi^{v,e}$ and $i\in [N]$, there is at most one
$j\in \pi^{-1}(i)$ such that $\frac{(l_v^{(i)})^2}{\|\vec{l}_v\|_2^2}\geq \frac{1}{d^8}$.

Based on the above property, we will show
\[
        \sum_{i \in [N]} \Big( \sum_{j \in \pi^{-1}(i)} |l_{v}^{(j)}|
        \Big)^4 \leq 2\tau \|\vec{l}_v\|^4_{2} \ .
\]

Notice that
        \begin{equation}\label{eqn:smeq}
                \sum _{i\in [N]}\big(\sum_{j \in \pi^{-1}(i)} |l_{v}^{(j)}|\big)^4
                =  \sum_{i\in [N]} \sum_{j_1,j_2,j_3,j_4\in \pi^{-1}(i)} \big |l_v^{(j_1)}l_v^{(j_2)}l_v^{(j_3)}l_v^{(j_4)}\big|
        \end{equation}
and the sum of all the terms with $j_1=j_2= j_3 =j_4$ is  $\|\vec{l}_v\|_4^4$.

For all other terms $|l_v^{(j_1)}l_v^{(j_2)}l_v^{(j_3)}l_v^{(j_4)}\big|$ with $j_1,j_2,j_3,j_4$ that are not all equal, there is at least one $|l_v^{(j_r)}|$ ($r\in [4]$) smaller than $\frac{\|\vec{l}_{v}\|_2}{d^4}$.
Therefore, $|l_v^{(j_1)}l_v^{(j_2)}l_v^{(j_3)}l_v^{(j_4)}\big|$  can be bounded by $$\frac{\|\vec{l}_{v}\|_2}{d^4}\big( \sum_{j_1,j_2,j_3,j_4} |l_v^{(j_1)}|^3+|l_v^{(j_2)}|^3+|l_v^{(j_3)}|^3+|l_v^{(j_4)}|^3\big).$$
Overall, expression (\ref{eqn:smeq}) can be bounded by
\begin{align*}
                & \|\vec{l}_v\|_4^4 +  \frac{\|\vec{l}_{v}\|_2}{d^4} \sum_{i \in [N]} \sum_{j_1,j_2,j_3,j_4 \in \pi^{-1}(i)} |l_v^{(j_1)}|^3+|l_v^{(j_2)}|^3+|l_v^{(j_3)}|^3+|l_v^{(j_4)}|^3
                                \\\leq & \tau^2 \|\vec{l}_v\|_2^4 +  \frac{\|l_{v}\|_2}{d^4} 4{d^3} \sum_{j \in [M]} |l_v^{(j)}|^3  \qquad \qquad \qquad \text{(since $|\pi^{-1}(i)| \leq d$, each $l_v^{(j)}$ appears at most $4d^3$ times)}
                                \\ \leq & (\tau^2 + 4\frac{\tau}{d})\|\vec{l}_v\|_2^4
                   \qquad \qquad \qquad \qquad  \text{($\vec{l}_v$ is $\tau$-regular
                   vector, so $|l_v^{j}| \leq \tau \|\vec{l}_v\|_2$ for all $j \in [M]$ )}
                                 \\ \leq & 2\tau \|\vec{l}_v\|_2^4.
        \end{align*}

\end{proof}

Let us fix a $2\tau$-nice hyperedge $e =
(v_1,\ldots,v_k)$.  As before let $\calE_e$ denote the distribution on examples restricted to those generated for hyperedge $e$.
We will analyze the probability that the halfspace $h(\vec{y})$ agrees with a random example from $\calE_e$.

Let $\pi^{v_1,e},\pi^{v_2,e},\ldots,\pi^{v_k,e} : [M] \to [N]$
denote the projections associated with the hyperedge $e$.
For the sake of brevity, we shall write
$\vec{w}_i,\vec{y}_i,\vec{l}_i$ instead of $\vec{w}_{v_i},
\vec{y}_{v_i},\vec{l}_{v_i}$.  For all $j \in [N]$ and $i \in [k]$, define
$$\bvec{y}{j}_i = \truncate(\vec{y}_i, (\pi^{v_i,e})^{-1}(j)).$$
Similarly, define vectors $\bvec{w}{j}_i,
\bvec{l}{j}_i$ and $\bvec{s}{j}_i$.

Notice that for every example $(\vec{y},b)$ in the support of $\calE_e$, $\vec{y}_v = \vec{0}$ for every vertex $v \notin e$.  Therefore,
on restricting to examples from $\calE_e$ we can write:
\begin{eqnarray*}
        h(\vec{y}) = \sgn\Big(\sum_{i \in [k]} \iprod{\vec{w}_i,\vec{y}_i}
        - \theta\Big).
\end{eqnarray*}

\subsubsection{Common Influential Variables {(generalization of Lemma~\ref{lem:commoninf}})}\label{sec:cinf}
        \begin{lemma}
                \label{lem:commoninf2}
        Let $h(\vec y)$ be a halfspace such that for all $i \neq j \in
        [k]$,  we have $\pi^{v_i,e}(C_{\tau}(\vec{w}_i)) \cap
        \pi^{v_j,e}(C_{\tau}(\vec{w}_j)) = \emptyset$.  Then
        \begin{equation}\label{eqn:lctobound}
            \Big| \Ex_{\calE_e}[h(\vec{y}) | b =
                0] -
                \Ex_{\calE_e}[{h}(\vec{y}) | b=1] \Big| \leq
                O\Big(\frac{1}{k}\Big).
        \end{equation}
\end{lemma}
\begin{proof}
Fix the following notation:
\begin{align*}
 \vec{y}_i^{\hh} = \truncate(\vec{y}_i,
        C_{\tau}(\vec{w}_i)) & &
        \vec{y}^{\hh} = \vec{y}_1^{\hh}, \vec{y}_2^{\hh},\ldots,\vec{y}_k^{\hh} \\
        \vec{s} = \vec{s}_1,\vec{s}_2,\ldots,\vec{s}_k & &  \vec{l} = \vec{l}_1,\vec{l}_2,\ldots,\vec{l}_k.
\end{align*}
We can rewrite the halfspace $h(\vec y)$ as
$ h(\vec{y}) = \sgn\Big(\iprod{
\vec{s},\vec{y}^{\hh} } +  \iprod{
\vec{l},\vec{y} }  - \theta\Big)$.
Let us first normalize the weights of $h(\vec{y})$ so that
$\sum_{i \in [k]} \|\vec{l}_i\|_2^2 = 1$.
Let us condition on a possible fixing of the vector $\vec{y}^{\hh}$.
Under this conditioning and also for $b = 0$, define the family of
ensembles $\cA =
\bvec{A}{1},\ldots,\bvec{A}{N}$ as follows:
\begin{eqnarray*}
        \bvec{A}{j} = \Big\{ y^{(\rr)}_i\ |\ i \in [k], \rr \in
        [M] \text{ such
        that } \pi^{v_i,e}(\rr) = j \text{ and } \rr \notin
        C_{\tau}(\vec{w}_i) \Big\}
\end{eqnarray*}
Similarly define the ensemble $\cB =
\bvec{B}{1},\ldots,\bvec{B}{N}$ for the conditioning $b=1$.  Now we shall apply the invariance
principle (Theorem \ref{thm:invariance}) to the ensembles $\cA,\cB$ and the linear
function $\vec{l}(\vec{y})$:
\begin{equation*}
        \vec{l}(\vec{y}) = \sum_{j \in [N]}
        \iprod{\bvec{l}{j},\bvec{y}{j}}.
\end{equation*}
As we prove in Claim \ref{cl:matchingmoments} below, the ensembles $\cA, \cB$ have
matching moments up to degree $3$.  Furthermore, by Lemma \ref{lem:spread}, the linear function
$\vec{l}$ and the ensembles $\cA$, $\cB$ satisfy the following spread
property:
        \begin{align*}
        & \Pr_{\cA} \Big[ l(\cA) \in [\theta' - \alpha,
        \theta'+\alpha] \Big] \leq c(\alpha)  & & \Pr_{\cB} \Big[ l(\cB) \in [\theta' - \alpha,
        \theta'+\alpha] \Big] \leq c(\alpha)
        \end{align*}
        for all $\theta' \in \R$, where $c(\alpha) =8\alpha k +4\tau
        k + 2e^{-\frac{1}{2k^4\tau^2}}$ (by setting
        $\gamma = \frac{1}{k^2}$ and $|b-a| =2\alpha$ in Lemma
        \ref{lem:spread}).

Using the invariance principle (Th.~\ref{thm:invariance}), this implies:
\begin{multline}\label{eqn:tb}
        \left| \Ex_{\cA}\left[\sgn\left( \iprod{
\vec{s},\vec{y}^{\hh} } + \sum_{j \in [N]} \iprod{
\bvec{l}{j},\bvec{A}{j} }  - \theta\right)| \vec{y}^{\hh} \right] -
\Ex_{\cB}\left[\sgn\left( \iprod{
\vec{s},\vec{y}^{\hh} } + \sum_{j \in [N]} \iprod{
\bvec{l}{j},\bvec{B}{j} }  - \theta\right)| \vec{y}^{\hh} \right]   \right| \\
\leq O(\frac{1}{\alpha^4})\sum_{j \in [N]}
                \|\bvec{l}{j}\|_1^4 + 2 c(\alpha).
\end{multline}
Take $\alpha$ to be $\frac{1}{k^2}$ and recall that $\tau= \frac{1}{k^{13}}$.     In Claim \ref{cl:lisregular1} below we show that $$\sum_{j \in [N]} \| \bvec{l}{j} \|_1^4 \leq 2\tau k^4 .$$ The above inequality holds for an arbitrary conditioning of the values of $\vec{y}^{\hh}$. Hence,
by averaging over all settings of $\vec{y}^{\hh}$ we prove \eqref{eqn:lctobound}.
\end{proof}

\begin{claim}\label{cl:matchingmoments}
        The ensembles $\cA$ and $\cB$ have matching moments up to degree
        $3$.
\end{claim}
Let us suppose for a moment that $\vec{y}$ was generated by setting $ y_{v_i}^{(j)} = x_i^{(\pi^{v_i,e}(j))}$, that is
without adding any noise. By Lemma \ref{lem:d01}, the first four moments of
random variable $\vec{y}$ conditioned on $b=0$ agree with
the first moments of
random variable $\vec{y}$ conditioned on $b=1$. As we showed in Observation \ref{obs:addnoise_tomoments}, even with noise, the first four moments of
$\vec{y}$ remain the same when conditioned on $b=0$ and $b=1$.  Finally,
$\pi^{v_i,e}(C_{\tau}(\vec{w}_i)) \cap \pi^{v_j,e}(C_{\tau}(\vec{w}_j)) = \emptyset$
for all $i \neq j \in [k]$.  Hence for each $j \in [N]$, conditioning on
$\vec{y}^{\hh}$ fixes bits in at most one row of $\bvec{A}{j}$.  Formally,
for every $j \in [N]$, there exists at most one $i \in [k]$ such that
$\bvec{y}{j}_i$ and $\vec{y}^{\hh}$ have shared variables. Therefore, by Lemma \ref{lem:match}, $\cA$ and $\cB$ have matching moments up to degree $3$.

\begin{claim} \label{cl:lisregular1}
        $$\sum_{j \in [N]} \| \bvec{l}{j} \|_1^4 \leq 2\tau k^4\ .$$
\end{claim}
\begin{proof}
        Since $\| \bvec{l}{j} \|_1 = \sum_{i \in [k]}
        \|\bvec{l}{j}_i\|_1$, we can write
        \begin{align} \label{eq:nicetailbound}
                \sum_{j \in [N]} \| \bvec{l}{j} \|_1^4 \leq    \sum_{j
                \in N} k^4 \Big(\sum_{i \in [k]}  \| \bvec{l}{j}_i \|_1^4
                \Big) =   k^4 \sum_{i \in [k]}  \Big( \sum_{j
                \in [N]} \| \bvec{l}{j}_i \|_1^4
                \Big).
        \end{align}
        As $e = (v_1,\ldots,v_k)$ is a $2\tau$-nice hyperedge, we have
        $\sum_{j \in [N]} \| \bvec{l}{j}_i \|_1^4 \leq 2\tau \|
        \vec{l}_i \|_2^4$.  By normalization of $\vec{l}$, we know
        $\sum_{i \in [k]} \|\vec{l}_i\|_2^2 = 1$.  Substituting this into inequality (\ref{eq:nicetailbound}) we get the claimed bound.
\end{proof}

\subsubsection{Bounding the Number of Influential Coordinates (generalization of Lemma~\ref{lem:numinf})}\label{sec:ninf}

\begin{lemma}
        \label{lem:strongnuminf}
        Given a halfspace $h(\vec{y}) = \sgn(\sum_{i \in [k]}
        \iprod{{\vec{w}}_i,\vec{y}_i} -
        \theta)$ and $\rr \in [k]$ such that
        $|C_{\tau}(\vec{w}_\rr)| \geq t$ for $t =   \frac{1}{\tau^2} ( \lceil 4 k^2  \ln (2k)\rceil\lceil 4 \ln(1/\tau)\rceil  + \ln(1/\tau) + 10\ln d) = O(k^{29})$,  define $\tilde{h}(\vec{y}) =
        \sgn(\sum_{i \in [k]} \iprod{\tilde{\vec{w}}_i,\vec{y}_i} -
        \tilde{\theta})$ as follows:
        \begin{itemize}
                \item   $\tilde{\vec{w}}_\rr =
        \truncate(\vec{w}_\rr, \bb_t(\vec{w}_\rr))$ and
        $\tilde{\vec{w}}_{i} =
        \vec{w}_{i}$ for all $i \neq \rr$.
\item   $\tilde{\theta} =
        \theta - \E[\iprod{\vec{a}_\rr ,\vec{y}_\rr} | b=0]$, for $\vec{a} = \vec{w} - \tilde{\vec{w}}$.
        \end{itemize}
        Then,
        \begin{eqnarray*}
                \Big| \Ex_{\calE_e}[\tilde{h}(\vec{y}) | b =
                0] -
                \Ex_{\calE_e}[{h}(\vec{y}) | b=0] \Big| \leq \frac{1}{k^2}, & &          \Big|
                \Ex_{\calE_e}[\tilde{h}(\vec{y})|b=1] -
                \Ex_{\calE_e}[{h}(\vec{y})|b=1] \Big| \leq
                \frac{1}{k^2}.
        \end{eqnarray*}
\end{lemma}
\begin{proof}

                It is easy to see that the matching moments condition implies that
                \[
                         \E_{\calE_e}[\iprod{\vec{a}_\rr ,\vec{y}_\rr} | b=0] =  \E_{\calE_e}[\iprod{\vec{a}_\rr ,\vec{y}_\rr} | b=1].
                \]
                Let us show the inequality for the case $b = 0$, the
        other inequality can be derived in an identical way. Let $\calE_{e,0}$ denote distribution $\calE_{e}$ conditioned on $b=0$. Without loss
        of generality, we may assume that $\rr=1$ and  $|w_{1}^{(1)}| \geq
        |w_{1}^{(2)}| \ldots \geq |w_{1}^{(M)}|$.  In particular,
        this implies $\bb_{t}(\vec{w}_1) = \{1,\ldots,t\}$.
        Define
        \begin{eqnarray*}
                \mu_{\rr} = \E_{\calE_{e,0}}[\iprod{\vec{a}_{\rr} ,
                \vec{y}_{\rr}}],  & &
                \bsup{\mu}{i}_\rr = \E_{\calE_{e,0}}[\iprod{\bvec{a}{i}_\rr ,                \bvec{y}{i}_\rr}].
        \end{eqnarray*}
            Let us set $T = \lceil 4k^2 \ln (2k) \rceil$ and define the subset $G = \{g_1,\ldots,g_T\}$ of $\bb_{t}(\vec{w}_1)$ as follows:
            $$G = \{g_i \ |\  g_i = 1 + i \lceil(4/\tau^2)
        \ln(1/\tau)\rceil,  0\leq i \leq T \}.$$
        Therefore, by Lemma~\ref{lem:st}, $|w_1^{(g_i)}|$ is a geometrically decreasing
        sequence such that $|w_1^{(g_{i+1})}|\leq |w_1^{(g_i)}|/3$.  Let $H = \bb_t(\vec{w}_1)\setminus G$.
        Fix the following notation:
        \begin{align*}
        \vec{w}_1^{G} = \truncate(\vec{w}_1,G), & & \vec{w}_1^{H} =
        \truncate(\vec{w}_1,H), & & \vec{w}_1^{>t} =
        \truncate(\vec{w}_1,\{t+1,\ldots,n\}).
        \end{align*}
        Similarly, define the vectors $\vec{y}_1^G,\vec{y}_1^H,
        \vec{y}_1^{>t}$.  By definition, we have $\vec{a}_1 = \vec{w}_1^{>t}$.
        Rewriting the halfspace functions $h(\vec{y}),
        \tilde{h}(\vec{y}) $ :
        \begin{align*}  h(\vec{y}) &= \sgn \Big(\sum_{i =2}^k
                \iprod{\vec{w}_i,\vec{y}_i} + \iprod{\vec{w}_1^{G} , \vec{y}_1^G} +
        \iprod{\vec{w}_1^{H} , \vec{y}_1^H} + \iprod{\vec{a}_1 ,
        \vec{y}_1^{>t}} -
        \theta  \Big), \\
        \tilde{h}(\vec{y}) &= \sgn \Big(\sum_{i =2}^k
        \iprod{\vec{w}_i,\vec{y}_i} + \iprod{\vec{w}_1^{G} , \vec{y}_1^G} +
        \iprod{\vec{w}_1^{H} , \vec{y}_1^H}
        +\mu_1- \theta  \Big).
        \end{align*}

        By Claim \ref{claim:chebyshev} below, with probability at most
        $\frac{1}{d}=\frac{1}{4^k}$, we have $|\iprod{\vec{a}_1,\vec{y}_1} - \mu_1|
        \geq d^4 \|\vec{a}_1\|_2$.  Suppose $|\iprod{\vec{a}_1,\vec{y}_1} - \mu_1|
        < d^4 \|\vec{a}_1\|_2$, then Claim
        \ref{claim:geometric} below gives $|\iprod{\vec{a}_1,\vec{y}_1} - \mu_1|
        < 1/d^6 |w_{1}^{(g_T)}| < \frac{1}{3} |w_{1}^{(g_T)}|$.  Thus, we can write
        \begin{eqnarray*}
        \Pr_{\calE_{e,0}} \Big[h(\vec{y}) \neq \tilde{h}(\vec{y})\Big]
        \leq \Pr_{\calE_{e,0}} \Big[ \iprod{\vec{w}_1^{G} , \vec{y}_1^G}  \in
        [\theta'-\frac{1}{3}|w_{1}^{(g_T)}|,
        \theta'+\frac{1}{3}|w_{1}^{(g_T)}|]  \Big] + \frac{1}{4^k}.
        \end{eqnarray*}
        where $\theta' = -\sum_{i =2}^k
        \iprod{\vec{w}_i,\vec{y}_i} -
        \iprod{\vec{w}_1^{H} , \vec{y}_1^H}
        -\mu_1+ \theta$.
         For any fixing of the value of $\theta' \in \R$,
        induces a certain distribution on $\vec{y}_1^G$.  However,
        the $\frac{1}{k^2}$ noise introduced in $\vec{y}_1^G$ is
        completely independent.  This corresponds to the setting of
        Lemma \ref{lem:geometricsmallball1}, and hence we can bound the above probability by $(1-1/(2k^2))^{T} + 1/4^k \leq (1-1/(2k^2))^{4k^2 \ln{(2k)}} + 1/4^k \leq 1/k^2$.
\end{proof}

\begin{claim} \label{claim:chebyshev}
                        $$\Pr_{\calE_{e,0}}\Big[|\iprod{\vec{a}_1,\vec{y}_1} - \mu_1| \geq d^4
                \|\vec{a}_1\|_2 \Big] \leq \frac{1}{d}.$$
\end{claim}
\begin{proof}
Write $[M]$ as the union of  disjoint sets $R_1 \cup R_2 \cup \cdots \cup R_{N}$ where $R_i = (\pi^{v_1,e})^{-1}(i)$.
Notice  every $R_i$ has size at most $d$,
therefore $$\Var_{\calE_{e,0}}\big(\iprod{\vec{a}_1,\vec{y}_1} \big)  = \sum_{i \in [N]} \Var_{\calE_{e,0}}\big(\iprod{\vec{a}_1^{R_i},\vec{y}_1^{R_i}} \big) \leq
\sum_{i \in [N]} d\|\vec{a}_{1}^{R_i}\|_2^2
= d\|\vec{a}_{1}\|^2_2 .$$

By applying Chebyshev's inequality (Th.~\ref{thm:chebyshev}), we have
\[
\Pr_{\calE_{e,0}}\left[|\iprod{\vec{a}_1,\vec{y}_1} - \mu_1| \geq d^4\|\vec{a}_1\|_2\right]\leq \frac{2d}{d^8} \leq \frac{1}{d}.
\]
\end{proof}

        \begin{claim} \label{claim:geometric}   By the choice of the parameters $T$ and $t$,
        $$\|\vec{a}_1\|_2 \leq \frac{1}{d^{10}} |w_{1}^{(g_T)}|  .$$
        \end{claim}
        \begin{proof}
        By Lemma \ref{lem:st},
 \[ |w_1^{(g_T)}|^2 \geq \frac{\tau}{(1-\tau^2)^{t-g_T})} \|a_1\|^2_2 \geq \frac{\tau}{(1-\tau^2)^{\frac{1}{\tau^2} (\ln(1/\tau) + 10\ln d) }} \|\vec{a}_1\|^2_2 \geq d^{10} \|\vec{a}_1\|_2^2.
\]
       \end{proof}

\subsubsection{Proof of Soundness}\label{sec:all}
Recall that we chose $\tau = 1/k^{13}$ and $t = O(k^{29})$.
\begin{lemma} \label{lem:mainsoundness}
Fix a hyperedge $e$ which is $2\tau$-nice.        If for all $i \neq j \in [k]$,
        $\pi^{v_i,e}\big(\bb_{t}(\vec{w}_i)\big) \cap
        \pi^{v_j,e}\big(\bb_{t}(\vec{w}_j)\big) = \emptyset$ then the probability that halfspace
        $h(\vec{y})$ agrees with a random example from $\calE_e$ is
        at most $\frac{1}{2}+ O(\frac{1}{k})$.

\end{lemma}
\begin{proof}
The proof is similar to the proof of Theorem \ref{thm:testsoundness}.
Define $I = \{\rr \mid C_{\tau}(\w_\rr) > t \}$. We divide the problem into
the following two cases.
\begin{enumerate}
        \item $I = \emptyset$; i.e., for all $i \in [k]$, $C_{\tau}(\w_i) \leq t$. Then for any $i \neq j \in [k]$,                  $\bb_{t}(\vec{w}_i) \cap
    \bb_{t}(\vec{w}_j) = \emptyset$ implies $C_{\tau}(\vec{w}_i) \cap
    C_{\tau}(\vec{w}_j) = \emptyset$. By Lemma \ref{lem:commoninf2}, we have
        \[
           \Big|  \Ex_{\calE_e}[{h}(\vec{y}) | b=0] -  \Ex_{\calE_e}[{h}(\vec{y}) | b=1] \Big| \leq
               O\Big(\frac{1}{k}\Big).
        \]
\item
$I\ne \emptyset$. Then for all $\rr\in I$, we set $\tilde{\vec{w}}_\rr
= \truncate(\vec{w}_\rr, \bb_t(\vec{w}_\rr))$ and define a new
halfspace $h'$ by replacing $\vec{w}_\rr$ with $\tilde{\vec{w}}_\rr$ in $h$. Since such replacements occur at most $k$ times and, by
Lemma \ref{lem:strongnuminf}, every replacement changes the output of the
halfspace on at most $\frac{1}{k^2}$ fraction of examples from $\calE_e$, we can bound the overall
change by $k\times \frac{1}{k^2} = \frac{1}{k}$. That  is
   \begin{eqnarray}\label{eqn:diff1}
            \Big| \Ex_{ \calE_{e,0}}[h'(\vec{y})] -
            \Ex_{ \calE_{e,0}}[{h}(\vec{y})] \Big| \leq
            {\frac{1}{k}}, & &          \Big|
            \Ex_{ \calE_{e,1}}[h'(\vec{y})] -
            \Ex_{\calE_{e,1}}[{h}(\vec{y})] \Big| \leq
            {\frac{1}{k}}.
    \end{eqnarray}
    For the halfspace $h'$ and for all $\rr \in [k]$, we have
    $|C_{\tau}(\tilde{\vec{w}}_\rr)| \leq t$, thus reducing to Case 1.
    Therefore
        \begin{equation}\label{eqn:diff2}
         \Big| \Ex_{\calE_{e,o}}[h'(\vec{y})] -
        \Ex_{\calE_{e,1}}[h'(\vec{y})] \Big| \leq
       O\Big(\frac{1}{k}\Big).
     \end{equation}

Combining (\ref{eqn:diff1}) and (\ref{eqn:diff2}), we get
\[
         \Big| \Ex_{\calE_{e,0}}[h(\vec{y})] -
      \Ex_{\calE_{e,1}}[h(\vec{y})] \Big| \leq
     O\Big(\frac{1}{k}\Big).
   \]
\end{enumerate}
In other words, the probability that halfspace
        $h(\vec{y})$ agrees with a random example from $\calE_e$ is
        at most $\frac{1}{2}+ O(\frac{1}{k})$.
\end{proof}
We first recall the soundness statement:
\begin{proposition}
        If $\cL$ is not a $2k^2 2^{-\gamma k}$-weakly satisfiable instance of smooth \klabelcover,
                then there is no halfspace that agrees with a random example from $\calE$ with probability more than                 $\frac{1}{2} + \frac{1}{\sqrt{k}}$.
\end{proposition}
\begin{proof}
        The proof is by contradiction. We can define the following
        labeling strategy: for each vertex $v$, uniformly randomly
        pick a label from $\bb_t(\w_{v})$. We know that the size of
        $\bb_t(\w_{v_i})$ is $t=O(k^{29})$.

Suppose there exists a halfspace
that agrees with a random example from $\calE$ with probability more than                 $\frac{1}{2} + \frac{1}{\sqrt{k}}$. Then by an averaging argument,
for at least $\frac{1}{2\sqrt{k}}$-fraction of the hyperedges $e$,
$h(\y)$ agrees with a random example from $\calE_e$ with probability
at least $\frac12 + \frac{1}{2\sqrt{k}}$. We refer to these edges as {\em good}.

Since there is at most $O(1/k)$-fraction of the hyperedges that are {\it not} $2\tau$-nice we know that at least
$\frac{1}{4\sqrt{k}}$-fraction of the hyperedges are $2\tau$-nice and {\em good}. By Lemma \ref{lem:mainsoundness}, for each $2\tau$-nice and {\em good} hyperedge $e$ there exist two vertices $v_i,v_j \in e$ such that $\pi^{v_i,e}(\bb_t(\vec{w}_{i}))$ and $\pi^{v_j,e}(\bb_t(\vec{w}_j))$ intersect. Then there is a $\frac{1}{t^2}$ probability that the labeling strategy we defined will weakly satisfy hyperedge $e$.

Overall this strategy is expected to weakly satisfy at least $\frac{1}{4\sqrt{k}}\frac{1}{t^2} =
\Omega(\frac{1}{k^{59}})$ fraction of the hyperedges.  This is a contradiction since $\cL$ is not $\frac{2k^2}{2^{\gamma k}}$-weakly satisfiable.
\end{proof}

\bibliographystyle{abbrv}
\bibliography{learning}

\appendix

\section*{\Large \textbf{Appendix}}
\section{Probabilistic Inequalities}
In the discussion below we will make use of the following well-known inequalities.
\begin{theorem} (Hoeffding's Inequality)\label{thm:hof}
                Let $x^{(1)},\ldots,x^{(n)}$ be independent real random variables such that $x^{(i)} \in [a^{(i)},b^{(i)}]$.  Then the sum of these variables $S = \sum_{i=1}^{n}x^{(i)} $ satisfies
                \[
                        \Pr[|S-\E[S]|\geq nt]\leq 2 e^{-\frac{n^2t^2}{\sum_{i=1}^{n}(b^{(i)}-a^{(i)})^2}} .
                \]
\end{theorem}

\begin{theorem} (Berry-Esseen Theorem)\label{thm:berry}
        Let $x_1,x_2,\ldots,x_n$ be i.i.d. random unbiased $\{-1,1\}$ variables. Also assume that $\sum_{i=1}^{n}c_i^2 = 1$ and $\max_{i}\{ |c_i|\}\leq \alpha$. Let $g$ denote a unit Gaussian variable $N(0,1)$. Then for any $t\in \R$,
\[
        \left|\Pr\left[\sum c_{i} x_{i} \leq t \right] - \Pr[g\leq t]\right| \leq \alpha .
\]
\end{theorem}
\begin{theorem}
 \label{thm:chebyshev}
 (Chebyshev's Inequality) Let $X$ be a random variable with expected value $u$ and variance $\sigma^2$. Then for any real number $t>0$,
        \[
                \Pr[|X-\mu| \geq t \cdot \sigma] \leq 1/t^2.
        \]
\end{theorem}

\section{Proof of Lemma \ref{lem:spread} }\label{sec:lem23}

Recall that each $y^{(i)}$ is generated by the following manner:
\begin{equation}\label{eqn:nd}
        y^{(i)} = \begin{cases}  x^{(i)}&  \text{ with
                          probability }  1- \gamma \\
                          \text{random bit} & \text{ with
                          probability } \gamma.
          \end{cases}
        \end{equation}
Let us define a random  vector  $\z\in \bits^n$ based on $\y$. For $\y$ generated, if $y^{(i)}$ is generated as a copy of $x^{(i)}$ in (\ref{eqn:nd}),
then $\z^{(i)} =0$; if $y^{(i)}$ is generated as a  random bit in (\ref{eqn:nd}),
then $z^{(i)}=1$. Let us write $S = \sum_{i=1}^{n}w^{(i)} y^{(i)}$.
Our proof is based on two claims.
\begin{claim}\label{lem:bigvar} For a $\tau$-regular vector $\vec{w}$, $\Pr[\sum_{i=1}^n |w^{(i)}|^2 z^{(i)}\geq \gamma/2] \geq 1-2e^{-\frac{\gamma^2}{2\tau^2}}.$
\end{claim}
\begin{claim}\label{lem:smallball} For a $\tau$-regular vector $\vec{w}$, given any $a'< b'\in \R$ and any
        fixing of $z^{(1)},z^{(2)},\ldots,z^{(n)}$, if $\sum_{i=1}^n (w^{(i)})^2 z^{(i)} = \sigma^2 > 0$, then
$\Pr[S \in [a',b']] \leq         \frac{2|b'-a'|}{\sigma}+\frac{2\tau}{\sigma}.$
\end{claim}
Given the above two claims are correct, define event $V$ to be $\{\sum_{i=1}^n(w^{(i)})^2 z^{(i)} \geq \frac{\gamma}{2}\}$ and use $\mathbf{1}_{[a,b]}(x):\R\to \{0,1\}$ to denote the indicator function of whether $x$ falls into interval $[a,b]$.
\[
\begin{split}\Pr[S\in [a,b]] = \Ex[\mathbf{1}_{[a,b]}(S)]  =\Pr[V] \Ex[\mathbf{1}_{[a, b]} (S)\mid V] + \Pr[\neg{V}] \Ex[\mathbf{1}_{[a, b]} (S) \mid \neg{V}]
\end{split}
\]
By Claim \ref{lem:bigvar},  \[\Pr[\neg{V}] \Ex[\mathbf{1}_{[a, b]} (S) \ |\  \neg{V}] \leq \Pr[\neg V]\leq  2e^{-\frac{\gamma^2}{2\tau^2}}.\] By Claim \ref{lem:smallball},
\[\Pr[V] \Ex[\mathbf{1}_{[a, b]} (S)\mid V] \leq  \frac{4(b-a)}{\sqrt{\gamma}} + \frac{4\tau}{\sqrt{\gamma}}.\]

Overall,
\[
        \Pr\left[S\in [a,b]\right] \leq \frac{4(b-a)}{\sqrt{\gamma}} +\frac{4\tau}{\sqrt{\gamma}} + 2e^{-\frac{\gamma^2}{2\tau^2}}.
\]

It remains to verify Claim \ref{lem:bigvar} and Claim \ref{lem:smallball}.

To prove Claim \ref{lem:bigvar}, we need to apply the Hoeffding's inequality (see Theorem \ref{thm:hof}).

Notice that  $(w^{(i)})^2 z^{(i)} \in [0,(w^{(i)})^2]$ and applying Hoeffding's Inequality, we know

        \[
                \Pr\left[\left|\sum_{i=1}^n(w^{(i)})^2 z^{(i)} -  \Ex\left[\sum_{i=1}^n(w^{(i)})^2 z^{(i)} \right]\right|\geq   nt \right] \leq 2e^{\frac{-2n^2t^2}{\sum_{i=1}^{n}(w^{(i)})^4}}.
        \]

We know $\Ex[\sum_{i=1}^n(w^{(i)})^2 z^{(i)} ] = \gamma  $ and  $\sum_{i=1}^{n}((w^{(i)})^2)^2 \leq \max_{i}\left\{(w^{(i)})^2\right\} \sum_{i=1}^n(w^{(i)})^2 \leq \tau^2$.  If we take $nt = {\gamma/2}$, we have

\[
        \Pr\left[\left|\sum_{i=1}^n(w^{(i)})^2 z^{(i)} - \gamma\right| \geq  \frac{\gamma}{2}\right] \leq 2 e^{-\frac{\gamma^2}{2\tau^2}}.
\]

Therefore, with probability at least $1-2 e^{-\frac{\gamma^2}{2\tau^2}}$, $\sum_{i=1}^n(w^{(i)})^2 z^{(i)} \geq \frac{\gamma}{2}$.

To prove  Claim \ref{lem:smallball}, we need use Berry-Esseen Theorem (See Theorem \ref{thm:berry}).
Let us split $S$ into two parts: $S' = \sum_{z_i=1} w_i y_i $ and $S'' = \sum_{z_i=0} w_i y_i $. Since $S = S'+S''$ and $S'$ is independent of $S''$, it suffices to  show that  $\Pr\left[S' \in [a',b']\right] \leq      \frac{2|b'-a'|}{\sqrt{\sigma}}+\frac{2\tau}{\sigma} $ for any $a',b'\in \R$.
Define $y'^{(i)} = 2y^{(i)}-1$ and note that $y'^{(i)}$ a $\{-1,1\}$ variable. By rewriting $S'$ using this definition, we have $$S' = \sum_{z^{(i)}=1} w^{(i)} y^{(i)} = \sum_{z^{(i)}=1} w^{(i)}\frac{1+y'^{(i)}}{2}.$$

Then \begin{equation}\label{eqn:br}\Pr\left[S' \in [a',b']\right] = \Pr\left[\sum_{z^{(i)} = 1} w^{(i)} y'^{(i)} \in [a'',b'']\right],
\end{equation}  where $a'' = 2a'-\sum_{z^{(i)} = 1} w^{(i)}$ and $b'' =  2b'-{\sum_{z^{(i)}=1} w^{(i)}} $.
We can further rewrite the above term as
 \begin{multline*}
\Pr\left[\sum_{z^{(i)}=1} w^{(i)}y'^{(i)}\leq b''\right] - \Pr\left[\sum_{z^{(i)}=1} w^{(i)}y'^{(i)}\leq a''\right]\\=
\Pr\left[\sum_{z^{(i)} = 1} \frac{w^{(i)} y'^{(i)}}{\sqrt{\sum_{z^{(i)}=1}(w^{(i)})^2}} \leq \frac{b''}{{\sqrt{\sum_{z^{(i)}=1}(w^{(i)})^2}}}\right] -\Pr\left[\sum_{z^{(i)} = 1} \frac{w^{(i)} y'^{(i)}}{ {\sqrt{\sum_{z^{(i)}=1}(w^{(i)})^2}}}\leq \frac{a''}{     {\sqrt{\sum_{z^{(i)}=1}(w^{(i)})^2}}}\right].\end{multline*}

We can now apply Berry-Esseen's theorem. Notice that for all the $i$ such that $z^{(i)}=1$, $y'^{(i)}$ is distributed as an independent unbiased random $\{-1,1\}$ variable. Also $\max_{z^{(i)} =1} \frac{|w^{(i)}|}{\sqrt{\sum_{z^{(i)}=1} (w^{(i)})^2}}\leq \frac{\tau}{\sqrt{\sum_{z^{(i)}=1} (w^{(i)})^2}}.$

By Berry-Esseen's theorem, we know that expression (\ref{eqn:br}) is bounded by $$\Pr\left[N(0,1) \leq \frac{b''}{{\sqrt{\sum_{z^{(i)}=1}(w^{(i)})^2}}}\right] -\Pr\left[N(0,1)\leq \frac{a''}{         {\sqrt{\sum_{z^{(i)}=1}(w^{(i)})^2}}}\right] +\frac{2\tau}{\sqrt{\sum_{z^{(i)}=1}(w^{(i)})^2}}.$$
Using the fact that a unit Gaussian variable falls in any interval of length $\lambda$ with probability at most $\lambda$ and noticing that  $b''-a'' = 2(b'-a')$, we can bound the above quantity by
\[
\frac{2|b'-a'|}{\sqrt{\sum_{z^{(i)}=1}(w^{(i)})^2}}+\frac{2\tau}{\sqrt{\sum_{z^{(i)}=1}(w^{(i)})^2}} = \frac{2|b-a|}{\sigma} + \frac{2\tau}{\sigma}.
\]

\section{Proof of Invariance Principle (Th.~\ref{thm:invariance})}\label{sec:invar}
We restate our version of the invariance principle here for convenience.
\paragraph{Theorem \ref{thm:invariance} restated}{ {(Invariance Principle)}
        Let $\cA = \{\bvec{A}{1},\ldots,\bvec{A}{R}\},
        \cB = \{\bvec{B}{1},\ldots,\bvec{B}{R}\}$ be families of ensembles of random
        variables with $\bvec{A}{i} =
        \{a^{(i)}_1,\ldots,a^{(i)}_{k_i}\}$ and $\bvec{B}{i} =
        \{b^{(i)}_1,\ldots,b^{(i)}_{k_i}\}$, satisfying the following
        properties:
        \begin{itemize}\itemsep=0ex
                \item For each $i \in [R]$, the random variables in
                        ensembles $(\bvec{A}{i},\bvec{B}{i})$ have
                        matching moments up to degree $3$.  Further
                        all the random variables in $\cA$ and $\cB$ are
                        bounded by $1$.
                \item The ensembles $\bvec{A}{i}$ are all independent
                        of each other, similarly the ensembles
                        $\bvec{B}{i}$ are independent of each other.
        \end{itemize}
        Given a set of vectors $\vec{l} = \{\bvec{l}{1},\ldots,\bvec{l}{R}\}
        (\bvec{l}{i} \in \R^{k_i})$, define the linear function
        $\vec{l} :
        \R^{k_1} \times \cdots \times \R^{k_R} \to \R$ as
        $$\vec{l}(\vec{x}) = \sum_{i\in[R]} \iprod{
                \bvec{l}{i}, \bvec{x}{i}} $$
        Then for a $K$-bounded function $\Psi : \R
        \to \R$ we have
        \begin{equation}\label{eqn:smf1} \left|  \Ex_{\cA}\left[\Psi\Big(\vec{l}(\cA) -\theta \Big)\right] -
                \Ex_{\cB}\left[\Psi\Big( \vec{l}(\cB)
                -\theta \Big)\right] \right| \leq K \sum_{i \in [R]}
                \|\bvec{l}{i}\|_1^4 .\end{equation}
        for all $\theta > 0$. Further, define the spread function
        $c(\alpha)$ corresponding to the ensembles $\cA,\cB$ and the linear function
        $\vec{l}$ as follows,
        \begin{align*}
        (\textbf{Spread Function: }&)  \text{For  } 1/2> \alpha >0
        \text{, let  }\\
        c(\alpha)  = \max\big(& \sup_{\theta} \Pr_{\cA}  \Big[ \vec{l}(\cA) \in [\theta - \alpha,
        \theta+\alpha] \Big],  \quad \sup_{\theta} \Pr_{\cB} \Big[
        \vec{l}(\cB) \in [\theta - \alpha,
        \theta+\alpha] \Big]\big)
        \end{align*}
        then for all $\tilde{\theta}$,
        \begin{eqnarray}\label{eqn:sgn1}
                \Big|
                \Ex_{\cA}\left[\sgn\left(\vec{l}(\cA)-\tilde{\theta}\right)\right] -
                & \Ex_{\cB}\left[\sgn\left( \vec{l}(\cB)
                -\tilde{\theta}\right)\right]
                  \Bigg|  \leq
                   O\Big(\frac{1}{\alpha^4}\Big) \sum_{i \in [R]}
                \|\bvec{l}{i}\|_1^4 + 2c(\alpha  ).
        \end{eqnarray}
\begin{proof}
Let us prove equation (\ref{eqn:smf1}) first. Let $\calX_{i} = \{\bvec{B}{1},\ldots,\bvec{B}{i-1},\bvec{B}{i}, \bvec{A}{i+1},\ldots,\bvec{A}{R}\}$.

We know that
\begin{align*}
        \Ex_{\cA}[\Psi(\vec{l}(\cA) -\theta )] - \Ex_{\cB}[\Psi(\vec{l}(\cB) -\theta )]
        & =  \Ex_{\calX_{0}}[\Psi(\vec{l}(\calX_{0}) -\theta )]  - \Ex_{\calX_{R}}[\Psi(\vec{l}(\calX_{R}) -\theta )]
        \\& = \sum_{i=1}^{R}  \Ex_{\calX_{i-1}}[\Psi(\vec{l}(\calX_{i-1}) -\theta )] -    \Ex_{\calX_i}[\Psi(\vec{l}(\calX_i) -\theta )].
\end{align*}
Therefore, it suffices to prove  \begin{equation} \label{eqn:suf}\big|\Ex_{\calX_{i-1}}[\Psi(\vec{l}(\calX_{i-1}) -\theta )] - \Ex_{\calX_i}[\Psi(\vec{l}(\calX_i) -\theta )] \big|\leq K\|\bvec{l}{i}\|_1^4 . \end{equation}

Let $\calY_{i} = \{\bvec{B}{1},\ldots,\bvec{B}{i-1}, \bvec{A}{i+1},\ldots,\bvec{A}{R}\}$ and we have $\calX_{i} = \{\calY_{i},\bvec{B}{i}\}$ and $\calX_{i-1}=\{\calY_{i},\bvec{A}{i}\}$.
Then
\begin{multline}\label{eqn:bsum}
\Ex_{\calX_{i-1}}[\Psi(\vec{l}(\calX_{i-1}) -\theta )] -  \Ex_{\calX_{i}}[\Psi(\vec{l}(\calX_{i}) -\theta )]
=\Ex_{\calY_{i}}\left[\Ex_{\bvec{A}{i}}[\Psi(\vec{l}(\calX_{i-1}) -\theta )] -  \Ex_{\bvec{B}{i}}[\Psi(\vec{l}(\calX_{i}) -\theta )]\right].
\end{multline}

Notice that
\[
\vec{l}(\calX_{i-1})-\theta = \iprod{ \bvec{l}{i}, \bvec{A}{i}} + \sum_{1\leq j\leq  i-1} \iprod{ \bvec{l}{j}, \bvec{B}{j}}  +  \sum_{i+1\leq j\leq  R} \iprod{ \bvec{l}{j}, \bvec{A}{j}} -\theta
\]
and
\[
\vec{l}(\calX_{i})-\theta = \iprod{ \bvec{l}{i}, \bvec{B}{i}} + \sum_{1\leq j\leq  i-1} \iprod{ \bvec{l}{j}, \bvec{B}{j}}  +  \sum_{i+1\leq j\leq  R} \iprod{ \bvec{l}{j}, \bvec{A}{j}} -\theta.
\]

Take $\theta' =  \sum_{1\leq j\leq  i-1} \iprod{ \bvec{l}{j}, \bvec{B}{j}} +  \sum_{i+1\leq j\leq  R} \iprod{ \bvec{l}{j}, \bvec{A}{j}}  -\theta$,
We can further rewrite equation (\ref{eqn:bsum}) as
\begin{equation}\label{eqn:rep}
\Ex_{\calY_{i}}\big[\Ex_{\bvec{A}{i}}[\Psi(\iprod{ \bvec{l}{i}, \bvec{A}{i}} +\theta' )] -  \Ex_{\bvec{B}{i}}[\Psi(\iprod{ \bvec{l}{i}, \bvec{B}{i}} +\theta')]\big].
\end{equation}

Using the Taylor expansion of $\Psi$, we have that the inner expectation of equation (\ref{eqn:rep}) is equal to
\begin{align}\nonumber
         \big|\Ex_{\bvec{A}{i}}[\Psi(\theta') + \Psi'(\theta')  \iprod{ \bvec{l}{i}, \bvec{A}{i}} +   \frac{\Psi''(\theta')}{2} (\iprod{ \bvec{l}{i}, \bvec{A}{i}})^2
        + \frac{\Psi'''(\theta')}{6} (\iprod{ \bvec{l}{i}, \bvec{A}{i}})^3 + \frac{\Psi''''(\delta_1)}{24}  (\iprod{ \bvec{l}{i}, \bvec{A}{i}})^4] \\-
        \Ex_{\bvec{B}{i}}[\Psi(\theta') + \Psi'(\theta')  \iprod{ \bvec{l}{i}, \bvec{B}{i}} +   \frac{\Psi''(\theta')}{2} (\iprod{ \bvec{l}{i}, \bvec{B}{i}})^2
        + \frac{\Psi'''(\theta')}{6} (\iprod{ \bvec{l}{i}, \bvec{B}{i}})^3 + \frac{\Psi''''(\delta_2)}{24}  (\iprod{ \bvec{l}{i}, \bvec{B}{i}})^4] \big|.
        \label{eqn:inner}
\end{align}
for some $\delta_1,\delta_2 \in \R$.

Using the fact that $\bvec{A}{i}$ and $\bvec{B}{i}$ have matching moments up to degree $3$, we can upper bound equation (\ref{eqn:inner}) by
\[
        \left|\Ex_{\bvec{A}{i}}[\frac{\Psi''''(\delta_1)}{24}  (\iprod{ \bvec{l}{i}, \bvec{A}{i}})^4] -      \Ex_{\bvec{B}{i}}[\frac{\Psi''''(\delta_2)}{24}  (\iprod{ \bvec{l}{i}, \bvec{B}{i}})^4] \right|\leq \frac{K}{12} |\bvec{l}{i}|_1^4.\]
In the last inequality, we use the fact that $\Psi$ is $K$-bounded and $\iprod{ \bvec{l}{i}, \bvec{A}{i}} \leq \|\bvec{l}{i}\|_1, \iprod{ \bvec{l}{i}, \bvec{B}{i}} \leq \|\bvec{l}{i}\|_1$ since all random variables in $\cA,\cB$ are bounded by $1$.

Overall, we bound the inner expectation of equation (\ref{eqn:rep}) by  $\frac{K}{12} \|\bvec{l}{i}\|_1^4$. This implies equation (\ref{eqn:rep}) and therefore equation (\ref{eqn:suf}) is bounded by
 $\frac{K}{12} \|\bvec{l}{i}\|_1^4$, establishing equation (\ref{eqn:smf1}).

To prove equation (\ref{eqn:sgn1}), we need to use the following lemma.
\begin{lemma} (\cite{MosselOO05}, Lemma 3.21)
                There exists an absolute constant $C$ such that $\forall 0<\lambda<\frac12$, there exists  $\frac{C}{\lambda^4}$-bounded function $\Phi_{\lambda}:\R\to [0,1]$ which approximates the $\sgn(x)$ function in the following sense: $\Phi_{\lambda}(t) = 1$ for all $t>\lambda$; $\Phi_\lambda(t) = 0$ for  $t<-\lambda$.
\end{lemma}

By the above lemma, we can find a $\frac{C}{\alpha^4}$-bounded function $\Phi_{\alpha}$ such that  $\Phi_{\alpha}(\vec{l}(\cA)-\theta)$ is equal to
$\sgn(\vec{l}(\cA) -\theta)$ except when $\vec{l}(\cA)\in [\theta-\alpha,\theta+\alpha]$ and  $\Phi_{\alpha}(\vec{l}(\cB)-\theta)$ is equal to  $\sgn(\vec{l}(\cB) -\theta)$ except when $\vec{l}(\cB) \in [\theta-\alpha,\theta+\alpha]$. Also for any $x\in \R$, $|\sgn(x) - \Phi_\alpha(x)|\leq 1$ as $\sgn(x)$ and $\Phi_\alpha(x)$ are both in $[0,1]$.

Overall, we have
\begin{multline*}
        \left|  \Ex_{\cA}\left[\sgn\Big(\vec{l}(\cA)-\theta\Big)\right] -
        \Ex_{\cB}\left[\sgn\Big( \vec{l}(\cB) -\theta \Big)\right] \right| \leq \left|     \Ex_{\cA}\left[\sgn\Big(\vec{l}(\cA)-\theta\Big)\right] -       \Ex_{\cA}\left[\Phi_{\alpha}\Big(\vec{l}(\cA)-\theta\Big)\right] \right|\\ + \left | \Ex_{\cA}\left[\Phi_{\alpha}\Big(\vec{l}(\cA)-\theta\Big)\right]  - \Ex_{\cB}\left[\Phi_{\alpha}\Big( \vec{l}(\cB) -\theta \Big)\right]\right| + \left |\Ex_{\cB}\left[\Phi_{\alpha}\Big( \vec{l}(\cB) -\theta \Big)\right] -
        \Ex_{\cB}\left[\sgn\Big( \vec{l}(\cB) -\theta \Big)\right] \right| \\\leq \frac{C}{\alpha^4} \sum_{i \in [R]}
        \|\bvec{l}{i}\|_1^4 + 2 c(\alpha).
\end{multline*}
\end{proof}

\section{Hardness of Smooth \klabelcover}\label{sec:sm}
First we state the bipartite smooth Label Cover given by Khot \cite{Khot03}.
Our reduction is similar to the one in $\cite{GKS07J}$ but in addition requires proving the smoothness property.

\begin{definition} A Label Cover problem $\cL(G(W,V,E),M,N,\{\pi^{v,w}|(w,w)\in E\})$ consists of a  bipartite graph $G(V,W,E)$ with bipartition $V$ and $W$. $M,N$ are two positive integers such that $M>N$. There are  projection functions $\pi^{v,w}: [M]\to [N]$ associated with each edge $(w,v) \in E$ where $v\in V, w\in W$. All vertices in $W$ have the same degree (i.e., $W$-side regular). For any labeling $\Lambda: V\to [M]$ and $\Lambda: W\to [N]$, an edge is said to be satisfied if $\pi^{v,w}(\Lambda(v)) = \Lambda(w)$. We define $Opt(\cL)$ to be the maximum fraction of edges satisfied by any labeling.
\end{definition}
\begin{theorem}
There is an absolute constant $\gamma >0$ such that for all integer parameters $u$ and $J$, it is NP-hard to distinguish the following two cases: A Label Cover problem $\cL(G(W,V,E),N,M,\{\pi^{v,w} | (w,v)\in E\})$ with $M = 7^{(J+1)u}$ and $N= 2^{u}7^{Ju}$ having
        \begin{itemize}
                \item $Opt(\cL) = 1$ or
                \item $Opt(\cL)\leq 2^{-2\gamma u}$.
        \end{itemize}
In addition, the Label Cover has the following properties:
\begin{itemize}
        \item for each $\pi^{v,w}$ and any $i \in [N]$, we have $|(\pi^{v,w})^{-1}(i)|\leq 4^u$;
        \item for a fixed vertex $w$ and a randomly picked neighbor $v$ of $w$,
        \[
                \forall i,j\in [M], \Pr[\pi^{v,w}(i) = \pi^{v,w}(j)] \leq 1/J.
        \]

\end{itemize}
\end{theorem}
Below we prove Theorem \ref{thm:sml}.
\begin{proof}

Given an instance of bipartite Label Cover $\cL(G(V,W,E),M,N,\{\pi^{v,w}|(w,v)\in E\})$, we can convert it to a smooth \klabelcover instance $\cL'$ as  follows. The vertex set of $\cL'$ is $V$ and we generate the hyperedge set $E'$ and projections associated with the hyperedges in the following way:
\begin{enumerate}
        \item pick a vertex $w\in W$;
        \item pick a $k$-tuple of $w$'s neighbors $v_1,\ldots, v_k$ and add a hyperedge $e = (v_1,\ldots,v_k)$ to $E'$ with projections $\pi^{v_i,e} = \pi^{v_i,w}$ for each $i \in [k]$.
\end{enumerate}

\paragraph{Completeness:} If $Opt(\cL) = 1$, then there exists a labeling $\Lambda$ such that for every edge $(w,v)\in E$, $\pi^{v,w}(\Lambda(v)) =\Lambda(w)$. We can simply take the restriction of labeling $\Lambda$ on $V$ for the smooth \klabelcover instance $\cL'$.  For any hyperedge $e = (v_1,v_2,\ldots,v_k)$ generated by $w\in W$, we know $\pi^{v_i,e}(\Lambda(v_i)) = \Lambda(w) =\pi^{v_j,e}(\Lambda(v_j))$ for any $i,j\in [k]$.

\paragraph{Soundness:} If $Opt(\cL) \leq 2^{-2\gamma u}$, then we can weakly satisfy at most $2k^2 2^{-\gamma u}$-fraction of the hyperedges in $\cL'$.
This can be proved via contrapositive argument. Suppose there is a labeling strategy $\Lambda$ (defined on $V$) for the smooth \klabelcover that weakly satisfies $\alpha \geq 2k^2 2^{-\gamma u}$ fraction of the hyperedges. Extend the labelling to $W$ as follows: For each vertex $w \in W$ and a neighbor $v \in V$, let $\pi_{v,w}(\Lambda(v))$ be the label {\it recommended} by $v$ to $w$.    Simply assign for every vertex $w \in W$, the label most {\it recommended} by its neighbours.

By the fact that $\Lambda$ weakly satisfies $\alpha$-fraction of hyperedges in $\cL'$, we know that if we pick a vertex $w$ and randomly pick two of its neighbors $v_1,v_2$ then
\[
        \Pr\left[\pi^{v_1,w}(\Lambda(v_1)) = \pi^{v_2,w}(\Lambda(v_2))\right] \geq \frac{\alpha}{{k\choose 2}} \geq \frac{2\alpha}{k^2}.
\]
  By an averaging argument, at least $\frac{\alpha}{k^2}$-fraction of the vertices $w\in W$, will have the following property: among all the possible pairs of $w$'s neighbors, at least $\frac{\alpha}{k^2}$-fraction of pairs {\it recommend} the same label for $w$.  Let us call such a $w$ to be a {\it nice}.
It is easy to see that for every {\it nice} $w$, the most {\it recommended} label is actually recommended by at least $\frac{\alpha}{k^2}$ fraction of its neighbours.  Hence, the extended labelling satisfies at least $\alpha/k^2$ fraction of edges incident at each {\it nice} $w \in W$.  Using $W$-side regularity, we conclude that the extended labelling satisfies  $\frac{\alpha^2}{k^4} = 4 \cdot 2^{-2\gamma u}$-fraction the edges of $\cL$ -- a contradiction.

\paragraph{Smoothness of $\cL'$:} For any given vertex $v$ in $\cL'$, we want so show that if we randomly pick an hyperedge $e'$ containing $v$, then for the projection
$\pi^{v,e}$ as defined in $\cL'$,
\[
        \forall i,j\in [M], \Pr [\pi^{v,e'}(i) = \pi^{v,e'}(j)] \leq \frac{1}{J}.
\]

To see this, notice that all vertices in $W$ have the same degree; picking a
projection $\pi^{v,e'}$ using the above procedure is the same as randomly picking a neighbor $w$ of $v$ and using the projection $\pi^{v,w}$ defined in $\cL$. Therefore,
\[
      \forall i,j\in [M], \Pr[\pi^{v,e'}(i) = \pi^{v,e'}(j)  = \Pr [\pi^{v,w}(i) = \pi^{v,w}(j)] \leq \frac{1}{J}.
\]
\end{proof}

\end{document}